\documentclass[runningheads,envcountsame]{llncs}

\usepackage{amsmath,amssymb}
\usepackage{stmaryrd}

\usepackage[numbers,sort&compress]{natbib}

\usepackage[colorlinks,linkcolor=blue,citecolor=blue,urlcolor=blue]{hyperref}
\usepackage[capitalize,nameinlink]{cleveref}

\usepackage{tikz}
\tikzset{node/.style = {circle, fill, minimum size=4pt, inner sep=0pt}}

\newcommand{\st}{\,:\,}
\newcommand{\paths}{\mathcal{P}}

\newcommand{\LPMRF}{(P)}
\newcommand{\LPMRFdual}{(D)}

\newcommand{\mrf}{$\textup{MRF}$}
\newcommand{\rni}{$\textup{RNI}$}
\newcommand{\mrfdec}{$\textup{MRF}^{\star}$}
\newcommand{\rnidec}{$\textup{RNI}^{\star}$}

\newcommand{\mrfr}{$\textup{MRF-R}$}
\newcommand{\mrfrx}{$\textup{MRF-R}^{\star}$}
\newcommand{\mrfm}{$\textup{MRF-M}$}
\newcommand{\mrfmx}{$\textup{MRF-M}^{\star}$}

\newcommand{\fraccol}[1]{\chi_{\text{f}}(#1)}

\newcommand{\NP}{N\!P}
\newcommand{\coNP}{coN\!P}

\def\mathrlap{\mathpalette\mathrlapinternal} 
\def\mathrlapinternal#1#2{\rlap{$\mathsurround=0pt#1{#2}$}}

\newcommand{\elsum}[1]{\sum_{\mathrlap{#1}}\;}

\spnewtheorem{observation}[theorem]{Observation}{\bfseries}{\itshape}
\crefname{observation}{Observation}{Observations}
\Crefname{appsec}{Appendix}{Appendices}

\usepackage{thm-restate}

\title{Stronger Hardness for Maximum Robust Flow and Randomized Network Interdiction\vspace{-0.3cm}}
\titlerunning{Stronger Hardness for Maximum Robust Flow}

\author{Jannik Matuschke (KU Leuven)}
\authorrunning{Jannik Matuschke}
\institute{}

\begin{document}

\maketitle

\vspace*{-0.4cm}

\begin{abstract}
    We study the following fundamental network optimization problem known as \textsc{Maximum Robust Flow} (MRF):
    A planner determines a flow on $s$-$t$-paths in a given capacitated network.
    Then, an adversary removes $k$~arcs from the network, interrupting all flow on paths containing a removed arc.
    The planner's goal is to maximize the value of the surviving flow, anticipating the adversary's response (i.e., a worst-case failure of $k$ arcs).
    It has long been known that  MRF can be solved in polynomial time when $k = 1$ 
    \cite{AnejaChandrasekaranNair2001}, whereas it is $\NP$-hard when $k$ is part of the input~\cite{DisserMatuschke2016}.
    However, the complexity of the problem for constant values of $k > 1$ has remained~elusive, in part due to structure of the natural LP description preventing the use of the equivalence of optimization and separation.
    This paper introduces a reduction showing that the basic version of MRF described above encapsulates the seemingly much more general variant where the adversary's choices are constrained to $k$-cliques in a compatibility graph on the arcs of the network.
    As a consequence of this reduction, we are able to prove the following results: 
    (1) {\mrf} is $\NP$-hard for any constant number $k > 1$ of failing arcs, answering an open question raised in~\citep{DisserMatuschke2016}. 
    (2) When $k$ is part of the input, {\mrf} is~$P^{\NP[\log]}$-hard. 
    (3) The integer version of {\mrf} is $\Sigma_2^P$-hard.
\end{abstract}

\section{Introduction}
\label{sec:intro}

Network flows are a fundamental concept in combinatorial optimization and operations research. They form the central ingredient of numerous optimization models with applications concerning communication, energy transmission, logistics, road traffic, and many more.
In many of these applications, the underlying network infrastructure is susceptible to failures, e.g., due to technical malfunctions, unforeseen load peaks, or intentional sabotage,
causing flow using links affected by such failures to be interrupted.
This motivates the study of robust network optimization models~\cite{BertsimasNasrabadiStiller2013,bertsimas2013power,matuschke2016protecting,matuschke2020rerouting,DisserMatuschke2016,gottschalk2016robust,aggarwal2002multiroute,AnejaChandrasekaranNair2001,bertsimas2003robust}, which aim to produce resilient solutions by anticipating worst-case failure scenarios.

A basic model for robustness of network flows is the \textsc{Maximum Robust Flow}~(MRF) problem proposed by \citet*{AnejaChandrasekaranNair2001}.
This problem can be seen as a game played between a \emph{planner} and an \emph{interdictor}, an adversary emulating worst-case failure scenarios in the network.
The input to {\mrf} consists of a digraph $D = (V, A)$, arc capacities~\mbox{$u \in \mathbb{Q}_+^A$}, a source $s \in V$, a sink $t \in V$, and an interdiction budget $k \in \mathbb{N}$.
The planner wants to send flow from $s$ to $t$ with the goal of maximizing the amount surviving a worst-case failure of $k$ arcs selected by the interdictor, where any flow on paths containing a failed arc is lost.
Formally, let~$\mathcal{P} \subseteq 2^A$ denote the set of $s$-$t$-paths in $D$, let~$X := \{x \in \mathbb{Q}_+^{\paths} \st \sum_{P : a \in P} x_P \leq u_a\ \forall\; a \in A\}$
denote the set of $s$-$t$-flows in $D$ respecting the capacities $u$, and let $\mathcal{S}_k := \{S \subseteq A \st |S| = k\}$ denote the set of possible failure scenarios.
The planner then wants to solve
\begin{align*}
    \textstyle \max_{ x \in X} \min_{S \in \mathcal{S}_k} \sum_{P \in \paths : P \cap S = \emptyset} x_P.
\end{align*}

Note that the use of path flows in the problem definition is crucial, as the amount of flow interrupted by the failure of multiple arcs does not depend on the arc-flow values alone~(see, e.g.,~\citep[Fig.~6]{BertsimasNasrabadiStiller2013}). Note further that $|\paths|$ might be exponential in the size of the digraph $D$ given as input.

{\mrf} is closely related to the extensively studied \textsc{Network Interdiction} problem~\citep{chestnut2017hardness,zenklusen2010network,wood1993deterministic,guruganesh2014improved}, in which the players act in the reverse order, i.e., the interdictor removes a set of $k$ arcs so as to minimize the max-flow value in the remaining network.
\citet{bertsimas2013power} observed that the optimal value of {\mrf} equals that of \textsc{Randomized Network Interdiction} ({\rni}), in which the interdictor selects a probability distribution over possible failure sets and the planner chooses a flow knowing the distribution, i.e., 
\begin{align*}
    \textstyle \min_{z \in \Delta(\mathcal{S}_k)} \max_{ x \in X}   \sum_{P \in \paths} (1 - \sum_{S \in \mathcal{S}_k : P \cap S \neq \emptyset} z_S) \cdot x_P 
\end{align*}
where $\Delta(\mathcal{S}_k) := \{z \in [0, 1]^{\mathcal{S}_k} : \sum_{S \in \mathcal{S}_k} z_S = 1\}$ denotes the distributions over $\mathcal{S}_k$.

Indeed, {\mrf} and {\rni} can be equivalently described by the following pair of primal and dual linear programs (see \cref{app:interdiction} for a detailed discussion):

\hspace*{-0.5cm}%
\begin{minipage}{0.3\textwidth}
\begin{alignat*}{3}
    \LPMRF~\max \quad && \elsum{P \in \paths} x_P & \;-\; \lambda \\
    \text{s.t.} \quad && \elsum{P \in \paths : a \in P} x_P & \; \leq \; u_a && \quad \forall\; a \in A \\
    && \elsum{P \in \paths : P \cap S \neq \emptyset} x_P & \; \leq \; \lambda && \quad \forall\; S \in \mathcal{S}_k \\
    && x & \; \geq \; 0
\end{alignat*}
\end{minipage}
\begin{minipage}{0.55\textwidth}
\begin{alignat*}{3}
    \LPMRFdual~\min \quad && \elsum{a \in A} u_a y_a \qquad\quad \\
    \text{s.t.} \quad && \elsum{a \in P} y_a + \elsum{S \in \mathcal{S}_k : P \cap S \neq \emptyset} z_S & \; \geq \; 1 && \quad \forall\; P \in \paths \\
    && \elsum{S \in \mathcal{S}_k} z_S & \; = \; 1 && \\
    && y, z & \; \geq \; 0
\end{alignat*}
\end{minipage}\\[5pt]
Note that, at optimum, variable $\lambda$ equals the amount of flow lost in a worst-case scenario.
However, both $\LPMRF$ and $\LPMRFdual$ feature an exponential number of~variables.

In this paper, we study the complexity of {\mrf} and {\rni}.
For this purpose, we will refer to the corresponding decision versions of the problems:
\begin{quote}
\textbf{\mrfdec:} Given an instance $(D = (V, A), s, t, u, k)$ of {\mrf} and $L \in \mathbb{Q}_+$, decide whether $\max_{x \in X} \min_{S \in \mathcal{S}_k} \sum_{P \in \paths : P \cap S = \emptyset} x_P \geq L$.\\[5pt]  
\textbf{\rnidec:} Given an instance $(D = (V, A), s, t, u, k)$ of {\mrf} and $L \in \mathbb{Q}_+$, decide whether \mbox{$\min_{z \in \Delta(\mathcal{S}_k)} \max_{ x \in X} \sum_{P \in \paths} (1 - \sum_{S \in \mathcal{S}_k : P \cap S \neq \emptyset} z_S) \cdot x_P < L$}.
\end{quote}

Note that {\mrfdec} and {\rnidec} are complements of one another.
Moreover, for constant values of $k$, the number of constraints in $\LPMRF$ is polynomial, implying that it has an optimal solution with polynomial support (though the number of variables may still be exponential).
Hence, when restricting to instances where~$k$ is bounded by a constant, {\mrfdec} is in $\NP$ and {\rnidec} is in $\coNP$.
In general, however, {\mrfdec} and {\rnidec} are not known to be contained in $\NP$ or $\coNP$.

\subsection{Previous results on the complexity of {\mrf}/{\rni}}
\label{sec:related-work}

\citet{AnejaChandrasekaranNair2001} observed that for $k = 1$, the worst-case amount of flow lost in any solution to {\mrf} equals the maximum flow through any arc. Thus, {\mrf} and {\rni} can be solved in polynomial time by projecting $\LPMRF$ to arc-flow variables. 
Interestingly, by a simple rounding argument, this result extends to the case where the flow is required to be integral.

An alternative way to obtain the result from \cite{AnejaChandrasekaranNair2001} lies in solving the separation problem of $\LPMRFdual$, which for $k = 1$ reduces to a standard shortest path problem. 
However,
\citet{DuChandrasekaran2007} showed that the dual separation problem becomes $\NP$-hard already for $k = 2$. 
They suggested that this implies, via the equivalence of optimization and separation~\citep{grotschel2012geometric}, that {\mrf}/{\rni} for $k = 2$ are likewise $\NP$-hard.
Unfortunately, \citet{DisserMatuschke2016} later observed that this implication does not hold:
The equivalence of optimization and separation implies that it is $\NP$-hard to optimize \emph{arbitrary} linear objectives over the polyhedron defined by~$\LPMRFdual$. However, the linear objectives that arise from {\mrf}/{\rni} are not arbitrary, as the $z$-variables do not appear in the objective.
Indeed, the instances constructed in~\citep{DuChandrasekaran2007} to establish $\NP$-hardness of dual separation contain $s$-$t$-cuts of cardinality~$2$. 
For such instances, both $\LPMRF$ and $\LPMRFdual$  have trivial solutions of value~$0$. 
Therefore, the authors of~\citep{DisserMatuschke2016} conclude:
``If \textsc{Maximum Robust Flow} for $k = 2$ is indeed $\NP$-hard, a reduction that shows this would need to construct considerably more involved instances of the problem.''
   
For the case that~$k$ is part of the input, \citet{DisserMatuschke2016} show that {\mrf} is strongly $\NP$-hard.
Their result relies heavily on the hardness of the interdictor's best-response problem, i.e., solving $\max_{S \in \mathcal{S}_k} \sum_{P : P \cap S \neq \emptyset} x_P$ for
given $x \in X$, which however can be solved by enumeration \mbox{when~$k$ is constant.}

See \cref{app:related-work} for a discussion of the approximability of {\mrf}/{\rni} and work on related problems, such as \textsc{Network Interdiction} and variants of {\mrf}/{\rni}.

\subsection{Our contribution}

We introduce a generalization of {\mrf} 
that we show to be equivalent to the original problem via a series of careful reductions.
This facilitates the analysis of {\mrf}/{\rni}, leading to a series of new strong hardness results~for~both~problems. 

\subsubsection*{Hardness results for {\mrf}/{\rni}.}

As mentioned above, the complexity of {\mrf}/{\rni} for constant values of $k > 1$ is an open question. 
We answer this question by showing that the problems are in fact hard already when~\mbox{$k = 2$}.

\begin{theorem}\label{thm:mrf-2-np}
    Let $k \in \mathbb{N}$ with $k > 1$. Then {\mrfdec} restricted to instances with interdiction budget $k$ is $\NP$-complete, and {\rnidec} restricted to instances with interdiction budget $k$ is $\coNP$-complete.
\end{theorem}

For the general case that $k$ is not restricted to a constant, we show that both problems are at least as hard as any problem in $P^{\NP[\log]}$, the class of  all decision problems that can be solved in polynomial time by an algorithm that makes a logarithmic number of calls to an $\NP$-oracle.%
\footnote{
An $\NP$-oracle is a subroutine that, given an instance of a problem from $\NP$, returns the correct answer in polynomial time.
The class $P^{\NP[\log]}$ (also known as $\theta_2^P$~\citep{buss1991truth,hemachandra1987strong}) is situated between $\NP$ and $\Sigma_2^P$ and contains the so-called boolean hierarchy~\citep{riege2006completeness}. 
}%

\begin{theorem}\label{thm:mrf-pnp}
    {\mrfdec} and {\rnidec} are $P^{\NP[\log]}$-hard.
\end{theorem}

\cref{thm:mrf-pnp} gives a strong indication that the two problems are neither contained in $\NP$ nor $\coNP$.
In fact, the result implies that it is not even possible to describe general instances of {\mrf} or {\rni} by boolean formulas over constantly many instances of problems from~$\NP$, unless the polynomial hierarchy collapses to its third level~\citep{chang1996boolean}.
From a practical point of view, this rules out, e.g., compact mixed integer programming formulations for the problems.
\cref{thm:mrf-2-np,thm:mrf-pnp} also reveal a surprising contrast between {\rni} and its integer version, \textsc{Network Interdiction} (NI).  
Note that a given solution $S \in \mathcal{S}_k$ to NI can be evaluated simply by computing the maximum flow value in the digraph~$(V, A \setminus S)$.
Hence, NI is in $\NP$, and for constant values of~$k$, it can be solved in polynomial time by enumeration.
Thus, {\rni} is an interesting case of a problem where the convexity introduced by allowing randomized strategies makes the problem harder rather than easier.

We also consider the integral version of {\mrf}, which requires the flow value on every path to be an integer. 
For this problem, we obtain hardness at the second level of the polynomial hierarchy.

\begin{theorem}\label{thm:mrf-sigma}
    Integral {\mrfdec} is $\Sigma_2^P$-hard.
\end{theorem}

We remark that \cref{thm:mrf-sigma} does not follow from recent meta theorems on the hardness of robust/interdiction versions of $\NP$-hard problems~\citep{grune2024complexity,grune2025completeness}. Indeed, as maximum integral $s$-$t$-flows can be computed in polynomial time, this problem marks a curious example of a problem whose complexity jumps at least by two levels in the polynomial hierarchy when going to the robust version.

\subsubsection*{Technical overview: {\mrf} with Restricted Interdiction.}
At the core of all three results mentioned above is a generalized version of {\mrf}, which we call \textsc{Maximum Robust Flow with Restricted Interdiction} ({\mrfr}), in which we are given an additional \emph{compatibility graph} $H = (A, E_H)$ on the arcs of $D$ and where the interdictor's choices are restricted to $$\mathcal{S}_{H,k} := \{S \subseteq A \st |S| \leq k, S \text{ is a clique in } H\}.$$
For our purposes, it will suffice consider the following decision variant of the problem on directed acyclic graphs (DAGs) with unit capacities:
\begin{quote}
\textbf{\mrfrx:} Given an instance $(D = (V, A), s, t, u, k)$ of {\mrf}, where $D$ is a DAG and $u \equiv 1$,
a compatibility graph $H = (A, E_H)$, and a demand~\mbox{$\theta \in \mathbb{Q}_{+}$}, decide whether there exists $x \in X$ such that \mbox{$\sum_{P \in \paths} x_P = \theta$} and $\max_{S \in \mathcal{S}_{H,k}} \sum_{P \in \paths : S \cap P \neq \emptyset} x_P \leq k - 1$.    
\end{quote}

Our main technical contribution is the following reduction, showing that any {\mrfrx} instance can in fact be modeled by an {\mrfdec} instance.

\begin{theorem}\label{thm:main-reduction}
    There is a polynomial transformation from {\mrfrx} to  {\mrfdec} that preserves the value of $k$.
    The same transformation also applies to the integral versions of the problems. The capacities in the resulting {\mrfdec} instance are polynomial in the size of the network.
\end{theorem}

Our proof of \cref{thm:main-reduction} consists of two major steps.
First, we show that we can model the compatibility graph in {\mrfrx} by introducing a multicommodity version of {\mrf} (called {\mrfmx}), where the commodities are used to ensure that incompatible arcs from the {\mrfrx} instance necessarily share a significant portion of flow on a common path, making it unattractive for the interdictor to simultaneously interdict any pair of such incompatible arcs.
We discuss this reduction from {\mrfrx} to {\mrfmx} in \cref{sec:mrfr-to-mrfm}.

Second, we show that the multicommodity problem {\mrfmx} can be reduced to the single-commodity problem {\mrfdec} by carefully building a ``wrapper'' around the multicommodity network that ensures that any optimal single-commodity flow has to respect the source-sink pairings of the commodities in order to keep the loss at worst-case interdiction within a certain bound. This reduction is discussed in \cref{sec:mrfm-to-mrf}.

Note that the reduction in \cref{thm:main-reduction} is robust in the sense that it preserves the value of $k$ as well as integrality of flows.
In order to establish the hardness results for {\mrf} discussed earlier, it thus suffices to show the same hardness results for {\mrfrx}.
Accordingly, we provide the following three reductions.

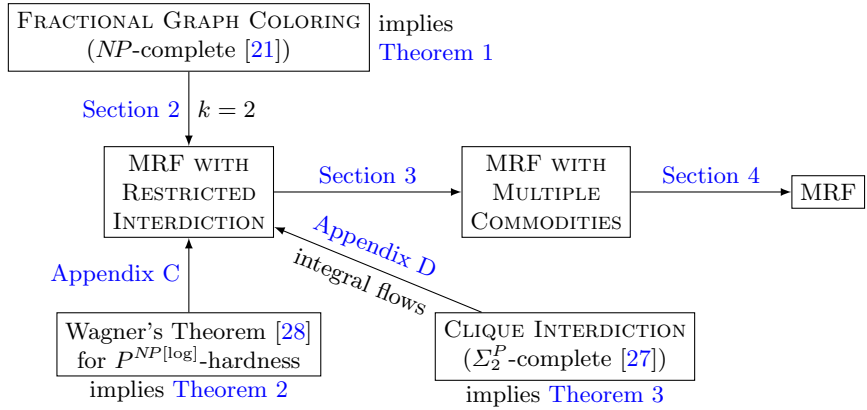
\begin{figure}[t]
    \centering

    \begin{tikzpicture}

        
        \path (0, 0)
            node[draw, align=center] (MRFR) {\textsc{{\mrf} with}\\ \textsc{Restricted}\\ \textsc{Interdiction}}
            ++(4.75, 0) node[draw, align=center] (MRFM) {\textsc{{\mrf} with}\\ \textsc{Multiple}\\ \textsc{Commodities}}
            ++(3.75, 0) node[draw, align=center] (MRF) {\mrf};

        \path (MRFR)
            +(0, 1.6) node[draw, align=center, anchor=south] (FracCol) {\textsc{Fractional Graph Coloring}\\ ($\NP$-complete~\citep{lund1994hardness})}
            +(0, -1.6) node[draw, align=center, anchor=north] (Wagner) {Wagner's Theorem~\citep{wagner1987more}\\ for $P^{\NP[\log]}$-hardness}
            +(5, -1.6) node[draw, align=center, anchor=north] (CliqueInterdiction) {\textsc{Clique Interdiction}\\
            ($\Sigma_2^P$-complete~\citep{rutenburg1994propositional})};

        \path (FracCol.east) ++(0, 0) node[anchor=west, align=left] {implies\\ \cref{thm:mrf-2-np}};
        
        \path (Wagner.south) ++(0, -0.2) node {implies \cref{thm:mrf-pnp}};

        \path (CliqueInterdiction.south) ++(0, -0.2) node {implies \cref{thm:mrf-sigma}};
        
        \draw[-latex] (MRFR) edge node[above, align=center] {\cref{sec:mrfr-to-mrfm}} (MRFM);

        \draw[-latex] (MRFM) edge node[above, align=center] {\cref{sec:mrfm-to-mrf}} (MRF);

        \draw[-latex] (CliqueInterdiction) edge 
        node[above, sloped, pos=0.55] {\cref{app:mrfr-sigma}}
        node[below, sloped, pos=0.55] {integral flows}
        (MRFR);

        \draw[-latex] (FracCol) edge 
        node[left] {\cref{sec:mrfr-np}} 
        node[right] {$k=2$} 
        (MRFR);

        \draw[-latex] (Wagner) edge node[left, align=center] 
        {\cref{app:mrfr-pnp}} (MRFR);
        
    \end{tikzpicture}

    \vspace*{-0.3cm}
    
    \caption{Overview of the reductions presented in this paper.}
    \vspace*{-0.3cm}
    \label{fig:reducions-overview}
\end{figure}

\begin{enumerate}
    \item To prove \cref{thm:mrf-2-np}, we provide a reduction from \textsc{Fractional Graph Coloring}, an $\NP$-complete LP relaxation of the classic \textsc{Graph Coloring} problem, to {\mrfrx} with $k = 2$.
    Here, we crucially use the flexibility won by the introduction of the compatibility graph to force the planner to send flow only along paths corresponding to feasible color classes.
    By the reduction from \textsc{Fractional Graph Coloring}, whose hardness stems from the PCP theorem~\citep{lund1994hardness}, our proof avoids the issue of the inapplicability of the equivalence of optimization and separation that has been a major obstacle in previous attempts to resolving the complexity of {\mrf} for constant $k$.    
    We discuss this reduction in \cref{sec:mrfr-np}.
    
    \item To prove \cref{thm:mrf-pnp}, we make use of \citeauthor{wagner1987more}'s Theorem~\citep{wagner1987more}, a meta theorem that establishes $P^{\NP[\log]}$-hardness by reducing from the question whether certain monotone sequences of instances of an $\NP$-complete problem contain an odd or even number of YES instances.
    To obtain our result, we reduce the even instances in the sequence to {\mrfrx} (using the preceding result), while we encode  the odd instances as \textsc{Clique} problems in the compatibility graph.
    In \cref{app:mrfr-pnp} we discuss how this can be done in such a way that a single  {\mrfrx} instance suffices to decide Wagner's parity problem.

    \item To prove \cref{thm:mrf-sigma}, we reduce the $\Sigma_2^P$-complete \textsc{Clique Interdiction} problem, which asks for a small subset of vertices to be removed from a graph $G$ so as to reduce its clique number, to \textsc{Integral \mrfrx}.
    The flexibility of {\mrfrx} allows us to encode the cliques of $G$ directly into the compatibility graph, and a careful construction forces the planner to send flow along a path that intersects with every sufficiently large clique in the compatibility graph.
    See \cref{app:mrfr-sigma} for details.
\end{enumerate}
A schematic overview of all reductions in this paper is given in \cref{fig:reducions-overview}.

\subsection{Notation}
\label{sec:preliminaries}

We introduce some notation that we will use throughout the paper. 
For $n \in \mathbb{N}$, we let $[n] := \{1, \dots, n\}$.
For a digraph $D = (V, A)$, we denote the set of arcs leaving/entering~$U \subseteq V$ by $\delta_D^+(U) := \{(v, w) \in A \st v \in U, w \in V \setminus U\}$ and $\delta_D^-(U) := \{(v, w) \in A \st v \in V \setminus U, w \in U\}$.
Finally, for any flow~$x \in X$, we use
\begin{align*}
    \textstyle x[a] := \sum_{P \in \paths : a \in P} x_P \quad\text{and}\quad \lambda(x, S) := \sum_{P \in \paths : P \cap S \neq \emptyset} x_P
\end{align*}
to denote the flow on arc $a \in A$ and the flow lost at failure of $S \subseteq A$.

\section{$\NP$-completeness of {\mrfrx} for $k = 2$}
\label{sec:mrfr-np}

In this section, we establish the following theorem, which together with \cref{thm:main-reduction} implies \cref{thm:mrf-2-np}. Its proof is based on a reduction from \textsc{Fractional Graph Coloring}, which introduce before describing the reduction.

\begin{theorem}\label{thm:mrfr-2-np}
    Let $k \in \mathbb{N}$ with $k > 1$. Then {\mrfrx} restricted to instances with interdiction budget $k$ is $\NP$-complete.
\end{theorem}

\subsubsection*{\textsc{Fractional graph coloring.}}
The \emph{fractional chromatic number} of an undirected graph $G = (V, E)$ is given by
\begin{alignat*}{3}
    \chi_{\text{f}}(G) := \min \quad && \textstyle \sum_{I \in \mathcal{I}(G)} y_I \quad\ \   \\
    \text{s.t.} \quad && \textstyle \sum_{I \in \mathcal{I}(G) : v \in I} y_I & \; = \; 1 && \quad \forall\; v \in V \\
    && y & \; \geq \; 0
\end{alignat*}
where $\mathcal{I}(G) \subseteq 2^V \setminus \{\emptyset\}$ is the set of independent sets in $G$. 
It is well-known that~$\chi_{\text{f}}(G)$ logarithmically approximates the chromatic number of $G$, from which it inherits $|V|^{-\delta}$-hardness of approximation for some constant $\delta > 0$~\citep{lund1994hardness}.
For our purposes the following consequence of this hardness suffices:
\begin{theorem}[{\citep[Theorem~2.9]{lund1994hardness}}]\label{thm:fractional-coloring-hardness}
    The following problem is $\NP$-complete: Given a graph $G$ and an integer $\ell \in \mathbb{N}$, decide whether $\fraccol{G} \leq \ell$.
\end{theorem}

\subsubsection{Constructing an {\mrfrx} instance.} 
Let~$G = (V, E)$ be an undirected graph and let $\ell \in \mathbb{N}$ (w.l.o.g.\ $\ell \geq 2$). 
We construct a digraph~$D = (V', A)$ with source $s$ and sink $t$ and a compatibility graph $H$.
Let $e_1, \dots, e_m$, with $m := |E|$, 
be an arbitrary ordering of the edges of $G$.
We define $V' := \{z_1, \dots, z_{m+1}\}$ and let $s := z_1$ and~\mbox{$t := z_{m+1}$}.
The arc set $A$ of $D$ is comprised by bundles of parallel arcs $A_i$ for each $i \in [m]$, where bundle $A_i$ contains $\ell$ arcs from $z_i$ to $z_{i+1}$.
For each bundle~$A_i$, representing edge $e_i = \{v, w\}$ in $G$, we fix two arbitrary but distinct arcs $a_{iv}$ and $a_{iw}$.
We further let $E_H := \big\{\{a_{iv}, a_{jv}\} \st i, j \in [m], i \neq j, v \in e_i \cap e_j\big\}$ be the edge set of the compatibility graph $H = (A, E_H)$, i.e.,~$H$ contains an edge between any two arcs representing the same vertex of~$G$. 
See \cref{fig:fractional-coloring-reduction} for a depiction of the construction.

\cref{thm:mrfr-2-np} follows from the following lemma, which shows that the {\mrfrx} instance with $k=2$ and $\theta = \ell$ is a YES instance if and only if $\chi_{\text{f}}(G) \leq \ell$.

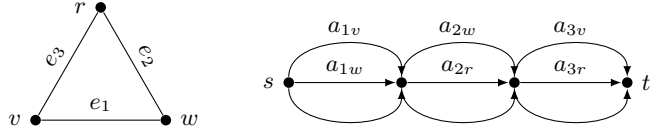
\begin{figure}[t]
    \centering

    \begin{tikzpicture}

        
        \node[node, label=left:{$v$}] (v1) at (210:1) {};
        \node[node, label=right:{$w$}] (v2) at (330:1) {};
        \node[node, label=left:{$r$}] (v3) at (90:1) {};
        
        \draw (v1) -- node[above] {$e_1$} (v2);
        \draw (v2) -- node[above, sloped] {$e_2$} (v3);
        \draw (v3) -- node[above, sloped] {$e_3$} (v1);

        
        \path (2.5, 0) node[node, label=left:{$s$}] (ze1) {}
            ++(1.5, 0)  node[node] (ze2) {}
            ++(1.5, 0)  node[node] (ze3) {}
            ++(1.5, 0)  node[node, label=right:{$t$}] (t) {};
        
        \draw[-latex, bend left=90] (ze1) edge node[above, sloped] {$a_{1v}$} (ze2);
        \draw[-latex] (ze1) edge node[above, sloped] {$a_{1w}$} (ze2);
        \draw[-latex, bend right=90] (ze1) edge (ze2);
        
        \draw[-latex, bend left=90] (ze2) edge node[above, sloped] {$a_{2w}$} (ze3);
        \draw[-latex] (ze2) edge node[above, sloped] {$a_{2r}$} (ze3);
        \draw[-latex, bend right=90] (ze2) edge (ze3);
        
        \draw[-latex, bend left=90] (ze3) edge node[above, sloped] {$a_{3v}$} (t);
        \draw[-latex] (ze3) edge node[above, sloped] {$a_{3r}$} (t);
        \draw[-latex, bend right=90] (ze3) edge (t);
    \end{tikzpicture}

    \vspace*{-0.2cm}
    
    \caption{Illustration of the reduction from \textsc{Fractional Graph Coloring} (graph on the left) to {\mrfrx} with $k = 2$ (digraph on the right). The compatibility graph of the constructed {\mrfrx} instance has the edges $\{a_{1v}, a_{3v}\}$, $\{a_{1w}, a_{2w}\}$, and $\{a_{2r}, a_{3r}\}$.}
    \label{fig:fractional-coloring-reduction}

    \vspace*{-0.4cm}
\end{figure}

\begin{lemma}\label{lem:coloring-path-all-vertices}
    There is an $s$-$t$-flow $x$ in $D$ for capacities $u \equiv 1$ with $\sum_{P \in \paths} x_P = \ell$ and $\max_{S \in \mathcal{S}_{H,2}} \lambda(x, S) \leq 1$ if and only if $\fraccol{G} \leq \ell$.
\end{lemma}

\begin{proof}
    We show how to construct a fractional coloring from a given flow $x$ with the properties stated in the lemma. 
    For the reverse direction see \cref{app:mrfr-np}.
    We start by defining
    $A_v := \{a \in A \st a = a_{iv} \text{ for some } i \in [m]\}$ for $v \in V$ and $I_P := \{v \in V \st A_v \subseteq P\}$ for $P \in \paths$.
    Note that $I_P \in \mathcal{I}(G)$ for every $P \in \paths$, as~$P$ cannot contain both parallel arcs $a_{iv}$ and $a_{iw}$ for any edge $e_i = \{v, w\} \in E$.
    
    Define $y_I := \sum_{P \in \paths : I_P = I} x_P$.
    Note that $\sum_{I \in \mathcal{I}(G) : v \in I} y_P  = \sum_{P \in \paths : v \in I_P} x_P$. We show that the latter equals $1$, which implies that $y$ is a fractional coloring of~$G$ and hence $\chi_{\text{f}}(G) \leq \sum_{I \in \mathcal{I}(G)} y_I = \sum_{P \in \paths} x_P = \ell$.

    Let $v \in V$ and fix any $a \in A_v$. 
    Observing that $x$ must saturate all arcs, we obtain $1 \geq \lambda(x, \{a, a'\}) = \sum_{P \in \paths : P \cap \{a, a'\} \neq \emptyset} x_P = 2 - \sum_{P \in \paths : a, a' \in P} x_P$ 
    for all~$a' \in A_v \setminus \{a\}$, where the first inequality follows from $\{a, a'\} \in E_H$ (as~$A_v$ is a clique in $H$).
    Rearranging yields $\sum_{P \in \paths : a, a' \in P} x_P = 1$ for all $a' \in A_v \setminus \{a\}$.
    Thus, if $a \in P$ for some $P \in \paths$ with $x_P > 0$, then $a' \in P$ for all $a' \in A_v$, or in other words: $v \in I_P$.
    Hence $\sum_{P \in \paths : v \in I_P} x_P = x[a] = 1$.
    \qed
\end{proof}

\section{Reducing {\mrfrx} to {\mrfmx}}
\label{sec:mrfr-to-mrfm}

For our reduction from {\mrfrx} to {\mrfdec}, we introduce a multi-commodity variant of {\mrf} ({\mrfm}) as an intermediate problem.
The input to this problem is a digraph $D = (V, A)$ with capacities $u$,  interdiction budget $k \in \mathbb{N}$, and a set of commodities $K$ with source $s_i$, sink $t_i$, and demand $d_i$ for each $i \in K$.
For~$i \in K$, let $\paths_i$ denote the set of $s_i$-$t_i$-paths in $D$ and let $\paths = \bigcup_{i \in K} \paths_i$.
We let 
\begin{align*}
    \textstyle X_K := \{x \in \mathbb{Q}_+^{\paths} \st \sum_{P \in \paths : a \in P} x_P \leq u_a\; \forall\, a \in A,\ \sum_{P \in \paths_i} x_P = d_i\; \forall\, i \in K\}
\end{align*}
denote the set of multi-commodity flows in $D$ for commodities $K$ that send $d_i$ units of flow from each $s_i$ to $t_i$ while respecting arc capacities $u$.
For our purposes it will again be useful to consider a particular decision variant of the problem:
\begin{quote}
\textbf{\mrfmx:} Given an input as described above, decide whether there is $x \in X_K$ with $\max_{S \in \mathcal{S}_k} \sum_{P \in \paths : P \cap S \neq \emptyset} x_P \leq kM - 1$, for $M := \max_{a \in A} u_a$.
\end{quote}

We show how to transform instances of {\mrfrx} to instances of {\mrfmx} with some additional properties that will be useful later when we continue the reduction in \cref{sec:mrfm-to-mrf} to finally obtain an {\mrfdec} instance.

\begin{restatable}{theorem}{restateThmMRFRtoMRFM}\label{thm:mrfr-to-mrfm}
    There is a polynomial transformation from {\mrfrx} to  {\mrfmx} that preserves the value of $k$.
    The same transformation also applies to the integral versions of the problems. 
    Moreover, the resulting {\mrfmx} instance has the following additional properties:
    \begin{enumerate}
    \item Capacities are integral and $M := \max_{a \in A} u_a$ is polynomial in $|A|$.
    \label{prop:capacities-M}
    \item All $i \in K$ fulfill $\sum_{a \in \delta^+(s_i)} u_a = \sum_{a \in \delta^-(t_i)} u_a = d_i$ and 
    $\delta^-(s_i) = \delta^+(t_i) = \emptyset$.\label{prop:commodity-source-sink}
    \item There is a commodity $i_0 \in K$ such that $\sum_{i \in K \setminus \{i_0\}} d_i \leq M - 3$.\label{prop:i0}
    \item For any $x \in X_K$ it holds that     
    $\sum_{P \in \mathcal{P}_i : a \in P} x_P \geq 2$ for all $i \in K \setminus \{i_0\}$ and all $a \in A$ with $u_a \geq M - 2$.\label{prop:commodity-arc}
    \item There is $U \subseteq V$ such that $\{s_i \st i \in K\} \subseteq U$, $\{t_i \st i \in K\} \subseteq V \setminus U$, 
    $\sum_{a \in \delta^+(U)} u_a = \sum_{i \in K} d_i$, and $|\{a \in \delta^+(U) \st u_a = M\}| = k - 1$.\label{prop:cut-M}
\end{enumerate}
\end{restatable}

\subsubsection*{Constructing an {\mrfmx} instance.}
To illustrate the main principle behind the reduction, we describe a simplified construction which ignores properties~\ref{prop:capacities-M} to \ref{prop:cut-M}.
The complete proof of \cref{thm:mrfr-to-mrfm}
can be found in~\cref{app:mrfr-to-mrfm}.
Given an instance of {\mrfrx} with DAG $D = (V, A)$, $s, t \in V$, interdiction budget $k$, demand $\theta$, and compatibility graph $H = (A, E_H)$, we construct an instance of {\mrfmx} with the same interdiction budget $k$ and a digraph $D' = (V', A')$ with capacities $u'$ and commodities $K$ as follows.

We first fix a linear order $\succ$ of the arcs of $D$ that is consistent with the DAG structure in the sense that for $a, a' \in A$ with $a \succ a'$ there is no path in $D$ from the tail of $a$ to the head of $a'$.
We define $F := \{(a, a') \st \{a, a'\} \notin E_H, a \succ a'\}$, i.e., $F$ contains all pairs of incompatible arcs, internally ordered by $\succ$. 
To construct $D'$, we subdivide every arc $a = (v, w)$ of $D$ by introducing the nodes $v^+_a$ and $v^-_a$ and replacing $a$ with the path $v$-$v^+_a$-$v^-_a$-$w$. We further introduce the arc $(v^-_a, v^+_{a'})$ for every $(a, a') \in F$.
We set the capacities $u'$ in $D'$ to $M:=|A|$ for the arcs of the form $(v^+_a, v^-_a)$ for $a \in A$, and to $1$ for all other arcs.
The set of commodities is $K := \{0\} \cup F \cup A$, where $s_0 = s$, $t_0 = t$, $d_0 = \theta$, $s_{a,a'} = v^+_a$, $t_{a,a'} = v^-_{a'}$, $d_{a,a'} = 1$, and $s_{a} = v^+_a$, $t_{a} = v^-_{a}$, $d_{a} = |\{a' \in A : \{a, a'\} \in E_H\}|$.
See \cref{fig:mrfx-mrfm-reduction-overview} for a depiction of the construction.

\begin{figure}[t]
    \centering

    \begin{tikzpicture}

        \path (0, 0) node[node, label=left:{$s$}] (s0) {}
                     ++(1.5, 0) node[node, label=above:{$v$}] (v0) {}
                     ++(1.5, 0) node[node, label=right:{$t$}] (t0) {};

        \draw[-latex] (s0) edge[bend left=30] node[above] {$a_1$} (v0);
        \draw[-latex] (s0) edge[bend right=30] node[below] {$a_2$} (v0);
        \draw[-latex] (v0) edge[bend left=30] node[above] {$a_3$} (t0);
        \draw[-latex] (v0) edge[bend right=30] node[below] {$a_4$} (t0);
        
        \path (4, 0) node[node, label=left:{$s$}] (s) {}
            ++(1, 0.5) node[node, label=above:{$v^+_{a_1}$}] (v11) {}
            +(0, -1) node[node, label=below:{$v^+_{a_2}$}] (v21) {}
            ++(1.5, 0) node[node, label=above:{$v^-_{a_1}$}] (v12) {}
            +(0, -1) node[node, label=below:{$v^-_{a_2}\;\;$}] (v22) {}
            ++(1, -0.5) node[node, label=above:{$v$}] (v) {}
            ++(1, 0.5) node[node, label=above:{$v^+_{a_3}$}] (v31) {}
            +(0, -1) node[node, label=below:{$v^+_{a_4}$}] (v41) {}
            ++(1.5, 0) node[node, label=above:{$v^-_{a_3}$}] (v32) {}
            +(0, -1) node[node, label=below:{$\;\;v^-_{a_4}$}] (v42) {}
            ++(1, -0.5) node[node, label=right:{$t$}] (t) {};

        \draw[-latex] (s) edge (v11);
        \draw[-latex] (s) edge (v21);
        \draw[-latex] (v11) edge (v12);
        \draw[-latex] (v21) edge (v22);
        \draw[-latex] (v12) edge (v);
        \draw[-latex] (v22) edge (v);
        
        \draw[-latex] (v42) edge (v31);
        \draw[-latex] (v42) edge[bend left=50] (v21);
        \draw[-latex] (v22) edge (v11);
        
        \draw[-latex] (v) edge (v31);
        \draw[-latex] (v) edge (v41);
        \draw[-latex] (v31) edge (v32);
        \draw[-latex] (v41) edge (v42);
        \draw[-latex] (v32) edge (t);
        \draw[-latex] (v42) edge (t);

        \draw[-latex, dotted] (v21) edge node[pos=0.66, fill=white, inner sep=0] {\scriptsize $1$} (v12);
        \draw[-latex, dotted] (v41) edge node[pos=0.66, fill=white, inner sep=0] {\scriptsize $1$} (v32);
        \draw[-latex, dotted] (v41) edge node[fill=white, midway, inner sep=0] {\scriptsize $1$} (v22);
        
        \draw[-latex, dotted, bend left=30] (v11) edge node[fill=white, midway, inner sep=0] {\scriptsize \raisebox{0.07cm}{$2$}} (v12);
        \draw[-latex, dotted, bend right=30] (v21) edge node[fill=white, midway, inner sep=0] {\scriptsize $1$} (v22);
        \draw[-latex, dotted, bend left=30] (v31) edge node[fill=white, midway, inner sep=0] {\scriptsize \raisebox{0.07cm}{$2$}} (v32);
        \draw[-latex, dotted, bend right=30] (v41) edge node[fill=white, midway, inner sep=0] {\scriptsize $1$} (v42);
        
    \end{tikzpicture}

    \vspace*{-0.6cm}
    
    \caption{From the \mbox{\mrfrx} instance with the digraph on the left and compatibility graph with edges $\{a_1, a_3\}, \{a_1, a_4\}, \{a_2, a_3\}$, we obtain the {\mrfmx} instance depicted on the right.
    The solid arcs represent the network $D'$, the dotted arcs indicate demands, i.e., a dotted arc $(s_i, t_i)$ represents a commodity~$i$ with source $s_i$, sink $t_i$, and demand~$d_i$ as indicated on the label  (we omit commodity $0$).
    Capacities are $1$ except for the four arcs $(v^+_{a_1}, v^-_{a_1}), \dots, (v^+_{a_4}, v^-_{a_4})$, each of which has capacity $M = 4$.}
    \label{fig:mrfx-mrfm-reduction-overview}

    \vspace*{-0.3cm}
\end{figure}
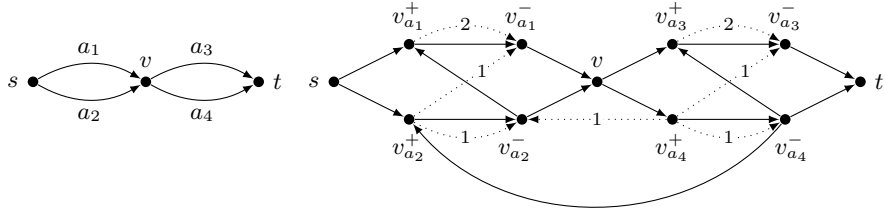

\subsubsection*{Analysis.}
Let $X_K$ denote the set of multicommodity flows in $D'$ for capacities~$u'$ fulfilling demands $d$ for commodities $K$ and let $\mathcal{S}'_k := \{S' \subseteq A' \st |S'| \leq k\}$.
The following lemma completes the proof of \cref{thm:mrfr-to-mrfm} (ignoring properties~\ref{prop:capacities-M} to \ref{prop:cut-M}).

\begin{lemma}
    There is $x' \in X_K$ with $\max_{'S \in \mathcal{S}'_k} \lambda(x', 'S) \leq kM - 1$
    if and only if there is $x \in X$ with $\sum_{P \in \paths} x_P = \theta$ and $\max_{S \in \mathcal{S}_{H,k}} \lambda(x, S) \leq k - 1$.
\end{lemma}
\begin{proof}
    Let $x \in X$ with $\sum_{P \in \paths} x_P = \theta$ and $\max_{S \in \mathcal{S}_{H,k}} \lambda(x, S) \leq k - 1$.
    Let $\phi(x)$ be the $s$-$t$-flow in $D'$ obtained from $x$ by replacing each arc~\mbox{$a = (v, w) \in A$} with its corresponding subpath $v$-$v^+_a$-$v^-_a$-$w$.
    Moreover, let $\hat{x}$ be the flow on commodities in $K \setminus \{0\}$ that routes the entire demand of commodity $a \in A$ along the one-arc path~$(v^+_a, v^-_a)$ and routes the entire demand of commodity $(a, a') \in F$ via the path~$v^+_{a}$-$v^-_{a}$-$v^+_{a'}$-$v^-_{a'}$. 
    Note that $\hat{x}[(v^+_a, v^-_a)] = M - 1$ for all $a \in A$ and that~\mbox{$\hat{x}[(v^-_a, v^+_{a'})] = 1$} for all $(a, a') \in F$.
    In particular, $\hat{x}$ leaves spare capacity exactly $1$ on each arc corresponding to the subdivision of $D$, while saturating all other arcs.
    Thus, $x' := \phi(x) + \hat{x} \in X_K$ is a feasible multicommodity flow (note that $\phi(x)[a] \leq 1$ for all $a \in A$ and $\phi(x)$ sends $\theta$ units from $s_0$ to $t_0$).
    
    Now assume by contradiction that there is $S' \in \mathcal{S}'_k$ with $\lambda(x', S') > kM - 1$.
    Note that $S'$ must consist of $k$ arcs of capacity $M$. Hence $S' = \{(v^+_a, v^-_a) \st a \in S\}$ for some $S \subseteq A$ with $|S| = k$.
    Moreover, if there is $(a, a') \in F$ with $a, a' \in S$, then $S'$ intersects the path $v^+_{a}$-$v^-_{a}$-$v^+_{a'}$-$v^-_{a'}$ of commodity $(a, a')$ twice, which would imply $\lambda(x', S') \leq kM - 1$.
    Thus, $\{a, a'\} \in E_H$ for all $a, a' \in S$ with~$a \neq a'$, which implies~$S \in \mathcal{S}_{H,k}$ and therefore $\lambda(x, S) \leq k - 1$.
    We obtain the contradiction $\lambda(x', S') = \lambda(\hat{x}, S') + \lambda(\phi(x), S') = k(M-1) + \lambda(x, S) \leq kM - 1$.

    Conversely, one can show that in fact any $x' \in X_K$ is of the form~\mbox{$\phi(x) + \hat{x}$} for some $x \in X$ with $\sum_{P \in \paths} x_P = \theta$ and that any $S \in \mathcal{S}_{H,k}$ with~$\lambda(x, S) > k - 1$ induces a set $S' \in \mathcal{S}'_k$ with~$\lambda(S', x') > kM - 1$. \qed
\end{proof}    

\section{Reducing {\mrfmx} to {\mrf}}
\label{sec:mrfm-to-mrf}

We prove the following theorem, which with \cref{thm:mrfr-to-mrfm} implies \cref{thm:main-reduction}.

\begin{theorem}\label{thm:mrfm-to-mrf}
    There is a polynomial transformation from {\mrfmx}, restricted to instances fulfilling properties~\ref{prop:capacities-M} to~\ref{prop:cut-M} from \cref{thm:mrfr-to-mrfm}, to {\mrfdec}. The transformation preserves the value of $k$.
    The same transformation also applies to the integral versions of the problems. 
\end{theorem}

\vspace*{-0.5cm}

\subsubsection*{Constructing an {\mrfdec} instance.}

    Consider an instance of {\mrfmx} given by a digraph $D = (V, A)$, capacities $u$ and commodities $K$ and interdiction budget~$k$, and that fulfills properties~\ref{prop:capacities-M} to~\ref{prop:cut-M} stated in \cref{thm:mrfr-to-mrfm}.
    For simplicity of notation, assume that $K = \{0, 1, \dots, m\}$ for $m = |K| - 1$, with $i_0 = 0$ for the commodity $i_0$ given by property~\ref{prop:i0}.

    We construct the following {\mrf} instance on digraph $D' = (V', A')$ with source~$s'$, sink~$t'$, capacities~$u'$, and the same interdiction budget~$k$.
    To simplify the description, we make use of \emph{immune arcs} that cannot be interdicted (this is w.l.o.g.\ in the context of this reduction;
    see \cref{app:mrfm-to-mrf-immune}).
    Define $K' := K \setminus \{0\}$ and let~$M' := M + 2k - 3$.
    The node set is 
    $V' := V \cup \{s', t', \bar{s}\} \cup \{s'_i \st i \in K'\}$ and the arc set $A'$ contains the following arcs:\\[-18pt]
    \begin{itemize}
        \item all arcs of $A$, with their original capacities $u$,
        \item immune arcs $(s', s_0)$ and $(t_0, t')$, each of capacity $d_0$,
        \item an immune arc $(s', s_1)$ of capacity $M'$,
        \item for each $i \in \{2, \dots, m\}$, an immune arc $(s', s_i)$, each of capacity $M' - 2k$,
        \item an immune arc $(s', \bar{s})$ of capacity $2k \cdot (M - 2)$,
        \item a bundle $\bar{A}$ of $2k$ arcs from $\bar{s}$ to $t'$, each with capacity $M - 1$;
        \item moreover, for every $i \in K'$:
        \begin{itemize}
            \item an arc $(s'_i, s_i)$ with capacity $M'$,
            \item an immune arc $(s_i, s'_{i+1})$ of capacity $2k$ (where $s'_{m+1} := \bar{s}$),
            \item for each $j \in K' \setminus \{i\}$, an immune arc $(s_i, t_j)$ of capacity $d_j$,
            \item a bundle $A_{i}$ of $d_i$ arcs from $t_i$ to $t'$, each with capacity $M - 1$,
            \item an immune arc $(s_i, t)$ of capacity $\xi^s_i := M' - 2k - \sum_{j \in K'} d_{j}$,
            \item an immune arc $(s', t_i)$ of capacity $\xi^t_i := d_i \cdot (M - 1 - m)$.
        \end{itemize}
    \end{itemize}
    The construction is depicted in \cref{fig:mrfm-mrf-reduction-overview}. Note that $\xi^s_i \geq 0$ due to property~\ref{prop:i0} and that with the exception of $s'$, $t'$, and the internal nodes of $D$, the in-capacity of a node equals its out-capacity.

\begin{figure}[t]
    \centering

    \begin{tikzpicture}

        \draw[fill=black!10!white] (4, 3) rectangle ++(5, -5.1);
        \draw (6.5, 2.75) node {$D$};
        
        \path node[node, label=left:{$s'$}] (s) {}
            ++(4.5, 2.7) node[node, label=below:{$s_0$}] (s0) {}
            ++(-2, -1.2) node[node, label=below:{$s'_1$}] (s01) {}
            +(2, 0)  node[node, label=below:{$s_1$}] (s1) {}
            ++(0, -1.5)  node[node, label=below:{$s'_2$}] (s02) {}
            +(2, 0)  node[node, label=below:{$s_2$}] (s2) {}
            ++(0, -1.5)  node[node, label=below:{$s'_3$}] (s03) {}
            +(2, 0)  node[node, label=below:{$s_3$}] (s3) {}
            ++(0, -1.5)  node[node, label=left:{$\bar{s}\;$}] (s04) {};

        \path (s1) ++(0.7, 0.9) node (s1p) {\scriptsize to $t'$};
        \draw[-latex, bend left=10] (s1) edge[dashed] node[pos=0.5, left] {\scriptsize  $\xi_1^s$} (s1p);
        \draw[-latex] (s2) edge[dashed] node[pos=0.25, left] {\scriptsize $\xi_2^s$} (s1p);
        \draw[-latex, bend right=15] (s3) edge[dashed] node[pos=0.15, left] {\scriptsize $\xi_3^s$} (s1p);

        \draw[-latex, bend left=35] (s) edge[dashed] node[above, sloped] {$d_0$} (s0);
        \draw[-latex] (s) edge[dashed] node[pos=0.6, above, sloped] {$M'$} (s01);
        \draw[-latex] (s) edge[dashed] node[pos=0.6, above, sloped] {$M' - 2k$} (s02);
        \draw[-latex] (s) edge[dashed] node[pos=0.6, above, sloped] {$M' - 2k$} (s03);
        \draw[-latex] (s) edge[dashed] node[pos=0.6, below, sloped] {$2k\cdot(M-2)$} (s04);
        
        \draw[-latex] (s01) edge node[above, sloped] {$M'$} (s1);
        \draw[-latex] (s02) edge node[above, sloped] {$M'$} (s2);
        \draw[-latex] (s03) edge node[above, sloped] {$M'$} (s3);
        
        \draw[-latex] (s1) edge[dashed] node[above, sloped] {$2k$} (s02);
        \draw[-latex] (s2) edge[dashed] node[above, sloped] {$2k$} (s03);
        \draw[-latex] (s3) edge[dashed] node[above, sloped] {$2k$} (s04);

        \path (s0) ++(4, 0) node[node, label=below:{$t_0$}] (t0) {}
            ++(0, -1.2)  node[node, label=below:{$t_1$}] (t1) {}
            ++(0, -1.5)  node[node, label=below:{$t_2$}] (t2) {}
            +(3, 0)  node[node, label=right:{$t'$}] (t) {}
            ++(0, -1.5)  node[node, label=below:{$t_3$}] (t3) {};

        \path (t1) ++(-0.9, 0.9) node (t1p) {\scriptsize from $s'$};
        \draw[-latex, bend left=10] (t1p) edge[dashed] node[pos=0.5, right] {\scriptsize $\xi_1^t$} (t1);
        \draw[-latex] (t1p) edge[dashed] node[pos=0.75, right] {\scriptsize $\xi_2^t$} (t2);
        \draw[-latex, bend right=5] (t1p) edge[dashed] node[pos=0.85, right] {\scriptsize $\xi_3^t$} (t3);

        \draw[-latex, bend left=35] (t0) edge[dashed] node[above, sloped] {$d_0$} (t);
        \draw[-latex] (t1) edge[double] node[above, sloped] {$d_1 \times (M - 1)$} (t);
        \draw[-latex] (t2) edge[double] node[above, sloped, pos=0.4] {$d_2 \; \times$} (t);
        \draw[-latex] (t2) edge[double] node[below, sloped, pos=0.4] {$(M - 1)$} (t);
        \draw[-latex] (t3) edge[double] node[sloped, below] {$d_3 \times (M - 1)$} (t);
        
        \draw[-latex] (s04) edge[double, out=0, in=250] node[above, sloped, pos=0.4] {$2k \times (M - 1)$} (t);

        \draw[-latex] (s1) edge[dashed] node[below, sloped, pos=0.8] {$d_2$} (t2);
        \draw[-latex] (s1) edge[dashed] node[below, pos=0.8] {$d_3$} (t3);

        \draw[-latex] (s2) edge[dashed, bend left=20] node[above, sloped, pos=0.7] {$d_1$} (t1);
        \draw[-latex] (s2) edge[dashed, bend right=20] node[below, sloped, pos=0.75] {$d_3$} (t3);

        \draw[-latex] (s3) edge[dashed] node[above, sloped, pos=0.75] {$d_1$} (t1);
        \draw[-latex] (s3) edge[dashed] node[below, sloped, pos=0.8] {$d_2$} (t2);
        
    \end{tikzpicture}

    \vspace*{-0.6cm}
    
    \caption{Illustration of the reduction from {\mrfmx} with digraph $D$ and four commodities to {\mrf}. 
    Dashed arcs represent immune arcs, with the label indicating the capacity.
    Single-lined solid arcs represent regular arcs, with the label indicating the capacity.
    Double-lined arcs labeled $\ell \times c$ represent bundles of $\ell$ arcs with capacity $c$ each.
    The figure omits the internal nodes and arcs of the digraph~$D$.}
    \label{fig:mrfm-mrf-reduction-overview}

    \vspace*{-0.5cm}
\end{figure}

\subsubsection*{Analysis.}
Let $\mathcal{S}_k := \{S \subseteq A \st |S| \leq k\}$ and $\mathcal{S}'_k := \{S \subseteq A' \setminus A_{\text{I}} \st |S| \leq k\}$, where $A_{\text{I}}$ is the set of immune arcs. 
Let~$X_K$ be the set of multicommodity flows in $D$ for capacities $u$, sending $d_i$ units from $s_i$ to $t_i$ for all $i \in K$. Let~$\paths'$ be the set of $s'$-$t'$-paths in $D'$ and let~$X'$ be the set of $s'$-$t'$-flows in $D'$ for capacities~$u'$.
To complete the reduction, we prove the following lemma, where~$\Delta := \sum_{a \in \delta_{D'}^+(s')} u'_a$.
\begin{lemma}\label{lem:equivalence-mflow-sflow}
    There is $x \in X_K$ with $\max_{S \in \mathcal{S}_k} \lambda(x, S) \leq kM - 1$ if and only if there is $x' \in X'$ with $\min_{S \in \mathcal{S}'_k} \sum_{P \in \paths' : P \cap S = \emptyset} x'_P \geq \Delta - (kM - 1)$.
\end{lemma}

We sketch the main ideas for proving the lemma. See \cref{app:mrfm-to-mrf} for the details. We construct the following ``base flow'' $\bar{x}$ that sends
\begin{itemize}
    \item $2k$ units of flow along paths $s'$-$s'_1$-$s_1$-$s'_2$-$s_2$-$\dots$-$s'_m$-$s_m$-$\bar{s}$-$t'$, 
    \item $2k \cdot (M-2)$ units of flow along paths $s'$-$\bar{s}$-$t'$,
    \item $\xi^s_i$ units of flow along the path $s'$-$s'_i$-$s_i$-$t'$ for each $i \in K'$,
    \item $\xi^t_i$ units of flow along paths $s'$-$t_i$-$t'$ for each $i \in K'$,
    \item $d_j$ units of flow along paths $s'$-$s'_i$-$s_i$-$t_j$-$t'$ for each $i, j \in K'$ with $i \neq j$,
\end{itemize}
where we split flow equally among the arcs of the bundle if the path described by the node sequence contains such a bundle.

Note that $\bar{x}$ sends $\Delta - \sum_{i \in K} d_i$ units of flow from $s'$ to $t'$ and that it does not use any arc in $A \cup \{(s', s_0), (t_0, t')\}$. Moreover, for $i \in K'$, it leaves a spare capacity of $d_i$ along the path $s'$-$s'_i$-$s_i$ and a spare capacity of~$1$ on each of the $d_i$ arcs in $A_i$.
We can thus combine $\bar{x}$ with a multicommodity flow $x \in X_K$, whose paths we appropriately extend to obtain an $s'$-$t'$-flow $\psi(x)$:
For every $P \in \paths_i$ and $a \in A_i$ for $i \in K'$, let $Q_{P, a}$ be the path derived from $P$ by attaching the path $s'$-$s'_i$-$s_i$ as a prefix and $a \in A_i$ as a suffix, and set $\psi(x)_{Q_{P,a}} := \frac{x_P}{d_i}$;
for every $P \in \paths_0$, set $\psi(x)_{R_P} := x_P$, where $R_P$ is the concatenation $(s', s_0) \circ P \circ (t_0, t')$.

\begin{restatable}{lemma}{restateLemMultiToSingle}\label{lem:multi-to-single}
    Let $x \in X_K$ with $\max_{S \in \mathcal{S}_k} \lambda(x, S) \leq kM - 1$ and let $x' := \psi(x) + \bar{x}$. Then $x' \in X'$ and $\min_{S \in \mathcal{S}'_k} \sum_{P \in \paths' : P \cap S = \emptyset} x'_P \geq \Delta - (kM - 1)$.
\end{restatable}
\begin{proof}[sketch]
    As argued above, $\psi(x)$ fits into the spare capacity left by $\bar{x}$ and thus $x' \in X'$ with $\sum_{P \in \paths'} x_P = \Delta$. 
    Moreover, we establish that $\lambda(x', S) \leq kM - 1$ for any $S \in \mathcal{S}'_k$.
    If $S \subseteq A$, this follows from $\lambda(x', S) = \lambda(x, S) \leq kM - 1$.
    If~$S \not\subseteq A$, then $\lambda(x', S) > kM - 1$ is only possible if $S$ contains at least one arc of capacity~$M'$, which must be of the form $(s'_i, s_i)$.
    For this case, one verifies that our construction ensures that~$(s'_i, s_i)$ shares a sufficiently high amount of flow on common paths with any other arc $a$ so that the marginal contribution of $a$ to $\lambda(x', S)$ is bounded by $M-2$ and hence $\lambda(x', S) \leq M' + (k-1)(M-2) = kM - 1$. 
    For $a \in A$, e.g., this follows from property~\ref{prop:commodity-arc} of \cref{thm:mrfr-to-mrfm}, while for $a \in A_i$ this follows from the fact that there is one unit of flow on paths containing both~$(s'_i, s_i)$ and $a$ as extensions of $s_i$-$t_i$-paths from $x$. \qed
\end{proof}

Conversely, one can show that every $x \in X'$ with the desired properties must actually be a combination of $\bar{x}$ with such an extended multicommodity flow. 
     
    \begin{restatable}{lemma}{restateLemHatX}\label{lem:hat-x}
        Let $x' \in X'$ with $\min_{S \in \mathcal{S}'_k} \sum_{P \in \paths' : P \cap S = \emptyset} x'_P \geq \Delta - (kM - 1)$.
        Define~$x$ by $x_P := \sum_{Q \in \paths' : P = Q \cap A} x'_Q$ for each $P \in \paths_i$ and $i \in K$.
        Then $x \in X_K$ and $\max_{S \in \mathcal{S}_k} \lambda(x, S) \leq kM - 1$.
    \end{restatable}

    \begin{proof}[sketch]
        We can show that $x'$ must saturate all arcs in $A' \setminus A$, as otherwise we can construct sets $S \in \mathcal{S}'_k$ whose failure leaves less than $\Delta - (kM-1)$ units of flow surviving (for constructing these sets we make use of a cut derived from property~\ref{prop:cut-M} of \cref{thm:mrfr-to-mrfm} and the bundle~$\bar{A}$).
        In fact, $x'$ must use the same paths as $\bar{x}$ on $A' \setminus A$.
        In addition, for each $a \in A_i$, the flow $x'$ must send one unit on paths containing both $(s'_i, s_i)$ and $a$, as otherwise $\lambda(x', S) > kM - 1$ for any set~$S$ containing $(s'_i, s_i)$, $a$, and $k-2$ arcs from $\bar{A}$.
        Because $\bar{x}$ does not send any flow on paths containing both $(s'_i, s_i)$ and $a$, these $|A_i| = d_i$ units of flow must traverse $s_i$-$t_i$-paths in $A$ for $i \in K'$.
        By a capacity argument, one obtains the same for $i = 0$.
        Thus $x \in X_K$, and $\lambda(x, S) \leq \lambda(x', S) \leq kM-1$ for all $S \in \mathcal{S}_k$.
        \qed    
    \end{proof}

\vspace*{-0.2cm}

\section{Conclusion}
\label{sec:conclusion}

Several interesting questions concerning {\mrf}/{\rni} remain.
Notably, the approximability of the problems is wide open:
All known approximation algorithms~\citep{BertsimasNasrabadiStiller2013,bertsimas2013power,baffier2016parametric} achieve approximation ratios linear in $k$, with no hardness of approximation results known (except for strong $\NP$-hardness ruling out fully polynomial-time approximation schemes). 
Our transformation from {\mrfrx} to {\mrfdec} might serve as a promising starting point for obtaining stronger inapproximability results. However, some significant extensions would be necessary so that the constructions can work for near-optimal instead of optimal solutions.

While our results indicate that {\mrfdec} and {\rnidec} are unlikey to be contained in $\NP$ or $\coNP$, pinpointing their exact complexity remains an intriguing question. At the moment, it is not even known whether the problems are contained in PSPACE. This question seems closely related to the size of the support of optimal solutions to the LPs~$\LPMRF$ and $\LPMRFdual$. Are there instances where optima necessarily have exponential support?

Finally, {\mrf}/{\rni} are an example of a combinatorial zero-sum game (where strategy sets are of exponential size and given only implicitly), and exponentially sized LPs arise frequently in this context. 
At the moment, showing hardness for such problems is challenging, and identifying more ``LP-type'' problems to reduce from (such as {\textsc{Fractional Coloring}, {\mrf}, and {\rni}) could facilitate the analysis of such games.

\subsubsection*{Acknowledgements.} The author thanks Martijn van Ee for pointing him to the complexity class $\Delta_2^P$, which led to the proof of \cref{thm:mrf-pnp}. 
This work was supported by the special research fund of KU Leuven (project~C14/22/026).

\bibliographystyle{splncs04nat}
\renewcommand{\bibsection}{\section*{References}}
\bibliography{references}

\begin{thebibliography}{30}
\providecommand{\natexlab}[1]{#1}
\providecommand{\url}[1]{\texttt{#1}}
\providecommand{\urlprefix}{URL }
\expandafter\ifx\csname urlstyle\endcsname\relax
  \providecommand{\doi}[1]{doi:\discretionary{}{}{}#1}\else
  \providecommand{\doi}{doi:\discretionary{}{}{}\begingroup
  \urlstyle{rm}\Url}\fi

\bibitem[{Aggarwal and Orlin(2002)}]{aggarwal2002multiroute}
Aggarwal, C.C., Orlin, J.B.: On multiroute maximum flows in networks. Networks
  \textbf{39}, 43--52 (2002)

\bibitem[{Aneja et~al.(2001)Aneja, Chandrasekaran, and
  Nair}]{AnejaChandrasekaranNair2001}
Aneja, Y.P., Chandrasekaran, R., Nair, K.P.K.: Maximizing residual flow under
  an arc destruction. Networks \textbf{38}, 194--198 (2001)

\bibitem[{Baffier et~al.(2016)Baffier, Suppakitpaisarn, Hiraishi, and
  Imai}]{baffier2016parametric}
Baffier, J.F., Suppakitpaisarn, V., Hiraishi, H., Imai, H.: Parametric
  multiroute flow and its application to multilink-attack network. Discrete
  Optimization \textbf{22}, 20--36 (2016)

\bibitem[{Bertsimas et~al.(2016)Bertsimas, Nasrabadi, and
  Orlin}]{bertsimas2013power}
Bertsimas, D., Nasrabadi, E., Orlin, J.B.: On the power of randomization in
  network interdiction. Operations Research Letters \textbf{44}, 114--120
  (2016)

\bibitem[{Bertsimas et~al.(2013)Bertsimas, Nasrabadi, and
  Stiller}]{BertsimasNasrabadiStiller2013}
Bertsimas, D., Nasrabadi, E., Stiller, S.: Robust and adaptive network flows.
  Operations Research \textbf{61}, 1218--1242 (2013)

\bibitem[{Bertsimas and Sim(2003)}]{bertsimas2003robust}
Bertsimas, D., Sim, M.: Robust discrete optimization and network flows.
  Mathematical Programming \textbf{98}, 49--71 (2003)

\bibitem[{Biefel et~al.(2025)Biefel, Kuchlbauer, Liers, and
  Waldm{\"u}ller}]{biefel2025robust}
Biefel, C., Kuchlbauer, M., Liers, F., Waldm{\"u}ller, L.: Robust static and
  dynamic maximum flows. Journal of Combinatorial Optimization \textbf{49}, 78
  (2025)

\bibitem[{Buss and Hay(1991)}]{buss1991truth}
Buss, S.R., Hay, L.: On truth-table reducibility to sat. Information and
  Computation \textbf{91}, 86--102 (1991)

\bibitem[{Chang and Kadin(1996)}]{chang1996boolean}
Chang, R., Kadin, J.: The boolean hierarchy and the polynomial hierarchy: A
  closer connection. SIAM Journal on Computing \textbf{25}, 340--354 (1996)

\bibitem[{Chestnut and Zenklusen(2017)}]{chestnut2017hardness}
Chestnut, S.R., Zenklusen, R.: Hardness and approximation for network flow
  interdiction. Networks \textbf{69}, 378--387 (2017)

\bibitem[{Dahan et~al.(2022)Dahan, Amin, and Jaillet}]{dahan2021probability}
Dahan, M., Amin, S., Jaillet, P.: Probability distributions on partially
  ordered sets and network interdiction games. Mathematics of Operations
  Research \textbf{47}, 458–484 (2022)

\bibitem[{Disser and Matuschke(2019)}]{DisserMatuschke2016}
Disser, Y., Matuschke, J.: The complexity of computing a robust flow.
  Operations Research Letters \textbf{48}, 18--23 (2019)

\bibitem[{Du and Chandrasekaran(2007)}]{DuChandrasekaran2007}
Du, D., Chandrasekaran, R.: The maximum residual flow problem: {NP}-hardness
  with two-arc destruction. Networks \textbf{50}, 181--182 (2007)

\bibitem[{Garey and Johnson(1979)}]{garey2002computers}
Garey, M.R., Johnson, D.S.: Computers and intractability. W.H. Freeman and
  Company (1979)

\bibitem[{Gottschalk et~al.(2018)Gottschalk, Koster, Liers, Peis, Schmand, and
  Wierz}]{gottschalk2016robust}
Gottschalk, C., Koster, A.M., Liers, F., Peis, B., Schmand, D., Wierz, A.:
  Robust flows over time: models and complexity results. {Mathematical
  Programming} \textbf{171}, 55--85 (2018)

\bibitem[{Gr{\"o}tschel et~al.(2012)Gr{\"o}tschel, Lov{\'a}sz, and
  Schrijver}]{grotschel2012geometric}
Gr{\"o}tschel, M., Lov{\'a}sz, L., Schrijver, A.: Geometric algorithms and
  combinatorial optimization, vol.~2. Springer Science \& Business Media (2012)

\bibitem[{Gr{\"u}ne and Wulf(2024)}]{grune2024complexity}
Gr{\"u}ne, C., Wulf, L.: On the complexity of recoverable robust optimization
  in the polynomial hierarchy. arXiv preprint arXiv:2411.18590  (2024)

\bibitem[{Gr{\"u}ne and Wulf(2025)}]{grune2025completeness}
Gr{\"u}ne, C., Wulf, L.: Completeness in the polynomial hierarchy for many
  natural problems in bilevel and robust optimization. In: Integer Programming
  and Combinatorial Optimization, Lecture Notes in Computer Science, vol.
  15620, pp. 256--269 (2025)

\bibitem[{Guruganesh et~al.(2014)Guruganesh, Sanita, and
  Swamy}]{guruganesh2014improved}
Guruganesh, G., Sanita, L., Swamy, C.: Improved region-growing and
  combinatorial algorithms for k-route cut problems. In: Proceedings of the
  twenty-sixth annual ACM-SIAM symposium on Discrete algorithms, pp. 676--695,
  SIAM (2014)

\bibitem[{Hemachandra(1987)}]{hemachandra1987strong}
Hemachandra, L.A.: The strong exponential hierarchy collapses. In: Proceedings
  of the nineteenth annual ACM symposium on Theory of computing, pp. 110--122
  (1987)

\bibitem[{Lund and Yannakakis(1994)}]{lund1994hardness}
Lund, C., Yannakakis, M.: On the hardness of approximating minimization
  problems. Journal of the ACM \textbf{41}, 960--981 (1994)

\bibitem[{Matuschke(2024)}]{matuschke2024decomposing}
Matuschke, J.: Decomposing probability marginals beyond affine requirements.
  In: Integer Programming and Combinatorial Optimization, Lecture Notes in
  Computer Science, vol. 14679, pp. 309--322, Springer (2024)

\bibitem[{Matuschke(2025)}]{matuschke2023decomposition-full}
Matuschke, J.: Decomposition of probability marginals for security games in
  max-flow/min-cut systems. Mathematical Programming \textbf{210}, 611--640
  (2025)

\bibitem[{Matuschke et~al.(2020)Matuschke, McCormick, and
  Oriolo}]{matuschke2020rerouting}
Matuschke, J., McCormick, S.T., Oriolo, G.: Rerouting flows when links fail.
  SIAM Journal on Discrete Mathematics \textbf{34}, 2082--2107 (2020)

\bibitem[{Matuschke et~al.(2017)Matuschke, McCormick, Oriolo, Peis, and
  Skutella}]{matuschke2016protecting}
Matuschke, J., McCormick, S.T., Oriolo, G., Peis, B., Skutella, M.: Protection
  of flows under targeted attacks. Operations Research Letters \textbf{45},
  53--59 (2017)

\bibitem[{Riege and Rothe(2006)}]{riege2006completeness}
Riege, T., Rothe, J.: Completeness in the boolean hierarchy:
  Exact-four-colorability, minimal graph uncolorability, and exact domatic
  number problems. In: Electron. Colloquium Comput. Complex. (2006)

\bibitem[{Rutenburg(1994)}]{rutenburg1994propositional}
Rutenburg, V.: Propositional truth maintenance systems: Classification and
  complexity analysis. Annals of Mathematics and Artificial Intelligence
  \textbf{10}, 207--231 (1994)

\bibitem[{Wagner(1987)}]{wagner1987more}
Wagner, K.W.: More complicated questions about maxima and minima, and some
  closures of {NP}. Theoretical Computer Science \textbf{51}, 53--80 (1987)

\bibitem[{Wood(1993)}]{wood1993deterministic}
Wood, R.K.: Deterministic network interdiction. Mathematical and Computer
  Modelling \textbf{17}, 1--18 (1993)

\bibitem[{Zenklusen(2010)}]{zenklusen2010network}
Zenklusen, R.: Network flow interdiction on planar graphs. Discrete Applied
  Mathematics \textbf{158}, 1441--1455 (2010)

\end{thebibliography}

\clearpage

\appendix

\crefalias{section}{appsec}
\crefalias{subsection}{appsec}
\crefalias{subsubsection}{appsec}

\section{Additional discussion for \cref{sec:intro}}

\subsection{The LP formulations $\LPMRF$ and $\LPMRFdual$}
\label{app:interdiction}

In this appendix, we have a closer look at the LP formulations $\LPMRF$ and $\LPMRFdual$, first introduced in~\citep{DuChandrasekaran2007}. In particular, we show that these LPs actually provide exact descriptions of {\mrf} and {\rni}, respectively.
For the purpose of this discussion, we restate both LPs.

\hspace*{-0.5cm}%
\begin{minipage}{0.3\textwidth}
\begin{alignat*}{3}
    \LPMRF~\max \quad && \elsum{P \in \paths} x_P & \;-\; \lambda \\
    \text{s.t.} \quad && \elsum{P \in \paths : a \in P} x_P & \; \leq \; u_a && \quad \forall\; a \in A \\
    && \elsum{P \in \paths : P \cap S \neq \emptyset} x_P & \; \leq \; \lambda && \quad \forall\; S \in \mathcal{S}_k \\
    && x & \; \geq \; 0
\end{alignat*}
\end{minipage}
\begin{minipage}{0.55\textwidth}
\begin{alignat*}{3}
    \LPMRFdual~\min \quad && \elsum{a \in A} u_a y_a \qquad\quad \\
    \text{s.t.} \quad && \elsum{a \in P} y_a + \elsum{S \in \mathcal{S}_k : P \cap S \neq \emptyset} z_S & \; \geq \; 1 && \quad \forall\; P \in \paths \\
    && \elsum{S \in \mathcal{S}_k} z_S & \; = \; 1 && \\
    && y, z & \; \geq \; 0
\end{alignat*}
\end{minipage}\\[5pt]

First, consider any optimal solution $(x^{\star}, \lambda^{\star})$ to $\LPMRF$. 
Note that $x^{\star} \in X$ by the first set of constraints.
Note further that $\lambda^{\star} = \max_{S \in \mathcal{S}_k} \sum_{P \in \paths : P \cap S \neq \emptyset} x^{\star}_P$, as otherwise $\lambda$ can be decreased, contradicting optimality.
Hence, the optimal value of {\LPMRF} equals
\begin{align*}
    \max_{x \in X} \sum_{P \in \paths} x_P - \max_{S \in \mathcal{S}_k} \ \ \ \elsum{\ P \in \paths : P \cap S \neq \emptyset} x_P \ = \ \max_{x \in X} \min_{S \in \mathcal{S}_k} \ \ \ \elsum{\ P \in \paths : P \cap S = \emptyset} x_P,
\end{align*}
i.e., the $x^{\star}$ is an optimal solution to {\mrf}.

Second, consider any optimal solution $(y^{\star}, z^{\star})$ to $\LPMRFdual$.
Note that $z^{\star} \in \Delta(\mathcal{S}_k)$ by feasibility.
Moreover, optimality implies that $y^{\star}$ must be an optimal solution to the following LP:
\begin{alignat*}{3}
    \min \quad && \elsum{a \in A} u_a & y_a \\
    \text{s.t.} \quad && \elsum{a \in P} y_a & \; \geq \; 1 - {\textstyle \sum_{S \in \mathcal{S}_k : P \cap S \neq \emptyset} z^{\star}_S} && \quad \forall\; P \in \paths \\
    && y & \; \geq \; 0
\end{alignat*}
By duality, we conclude that $\sum_{a \in A} u_a y^{\star}_a$  equals the optimal value of 
\begin{alignat*}{3}
    \max \quad && \elsum{P \in \paths} \big(1 - & {\textstyle \sum_{S \in \mathcal{S}_k : P \cap S \neq \emptyset} z^{\star}_S \big) \cdot x_P}\\
    \text{s.t.} \quad && \elsum{P \in \paths : a \in P} x_P & \; \leq \; u_a && \quad \forall\; a \in A \\
    && x & \; \geq \; 0
\end{alignat*}
and thus the optimal value of $\LPMRFdual$ equals
\begin{align*}
    \textstyle \min_{z \in \Delta(\mathcal{S}_k)} \max_{x \in X} \sum_{P \in \paths} (1 - \sum_{S \in \mathcal{S}_k : P \cap S \neq \emptyset} z_s) \cdot x_P,
\end{align*}
i.e., $z^{\star}$ is an optimal solution to {\rni}.

\subsubsection*{Number of variables.}
We further observe that $\LPMRF$ has $|\paths| + 1$ variables, which in general is exponential in the number of arcs $|A|$ of the underlying digraph.
Moreover, $\LPMRFdual$ has $|A| + \binom{|A|}{k} = \theta(|A|^k)$ variables.

\subsection{Further related work}
\label{app:related-work}

\subsubsection*{Approximation results.}
Various approximation algorithm have been devised for {\mrf} and {\rni}.
\citet{BertsimasNasrabadiStiller2013} use a restriction of $\LPMRF$ to obtain an approximation for {\mrf} whose factor depends on the proportion of flow lost in an optimal solution to.
\citet{bertsimas2013power} showed a guarantee of $\frac{k+1}{k+1+\lfloor k/2 \rfloor \lceil k/2 \rceil}$ for the same algorithm and provide an approximation for {\rni} with the same factor, using an elegant parametric primal-dual analysis.
\mbox{\citet{baffier2016parametric}} show that a different approach based on so-called $k$-route flows~\citep{aggarwal2002multiroute} results in a $\frac{1}{k+1}$-approximation.
No hardness-of-approximation results are known for the problem, except for the strong $\NP$-hardness of {\mrf}~\citep{DisserMatuschke2016} ruling out fully polynomial-time approximation schemes.

\subsubsection*{Network interdiction.}
\textsc{Network Interdiction} (also referred to as \textsc{Network Flow Interdiction} in literature) has been studied extensively. \mbox{\citet{wood1993deterministic}} showed that the problem is $\NP$-hard when $k$ is part of the input (whereas the problem can be solved by enumeration when $k$ is constant).
In some special cases such as planar graphs the problem can be solved in polynomial time~\citep{zenklusen2010network}.
However, the general case is also as hard to approximate as $k$-\textsc{Densest Subgraph}, for which no constant-factor approximations are known~\citep{guruganesh2014improved,chestnut2017hardness}.

\subsubsection*{Other related problems.}
Also numerous variants of robust flows, e.g., with possibility of rerouting~\citep{BertsimasNasrabadiStiller2013,matuschke2020rerouting}, the adversary attacking individual flow paths~\citep{matuschke2016protecting}, or flows over time~\citep{gottschalk2016robust,biefel2025robust} have been studied. 
A particularly closely related variant of {\mrf}/{\rni} was recently studied by~\citet{dahan2021probability}: 
In their security game, the interdictor is not restricted by a budget limiting the number of failing arcs, but interdicting an arc incurs a cost incorporated in the objective. 
This can be seen as a Lagrangian relaxation of the original problems.
\citet{dahan2021probability} showed that optimal primal and dual solutions can be found in polynomial time by projecting the interdictor's randomized strategies to probability marginals on the arcs, when the underlying digraph is acyclic.
This result was later extended to arbitrary digraphs and randomized interdiction on more general systems~\citep{matuschke2023decomposition-full,matuschke2024decomposing}.



\clearpage

\section{Missing proofs from \cref{sec:mrfr-np}}
\label{app:mrfr-np}

We show the missing direction from the proof of \cref{lem:coloring-path-all-vertices}.

\begin{restatable}{lemma}{restateLemColoringToFlow}\label{lem:coloring-to-flow}
    If $\fraccol{G} \leq \ell$, then there is an $s$-$t$-flow in $D$ for capacities $u \equiv 1$ with $\sum_{P \in \paths} x_P = \ell$ and $\max_{S \in \mathcal{S}_{H,2}} \lambda(x, S) \leq 1$.
\end{restatable}
\begin{proof}
    Let $y \in \mathbb{Q}_+^{\mathcal{I}(G)}$ be an optimal fractional coloring, i.e., $\sum_{I \in \mathcal{I}(G) : v \in I} y_I = 1$ for all $v \in V$ and $\sum_{I \in \mathcal{I}(G)} y_I = \fraccol{G} \leq \ell$.
    By setting $y_{\emptyset} := \ell - \sum_{I \in \mathcal{I}(G), I \neq \emptyset} y_I$ we can ensure that $\sum_{I \in \mathcal{I}(G)} y_I = \ell$ without loss of generality. 
    
    For $I \in \mathcal{I}(G)$ define $P_I :=  \{a_{iv} : v \in I, e_i \in \delta_G(v)\}$.
    Note that $|P_I \cap A'_i| \leq 1$ for all $i \in [m]$ because $I$ is an independent set in $G$.
    Let us further define $\bar{A}_i := A'_i \setminus \{a_{iv}, a_{iw}\}$, where $e_i = \{v, w\}$, and let $\mathcal{A}_I := \{\bar{A}_i \st I \cap e_i = \emptyset\}$.
    Moreover, let $\mathcal{Q}_I := \{Q \st |Q \cap A| = 1 \ \forall\, A \in \mathcal{A}_I\}$ denote the set of transversals of $\mathcal{A}_I$.
    Note that by construction $P_I \cup Q$ is an $s$-$t$-path in $D'$ for every $Q \in \mathcal{Q}_I$.
    
    We construct a flow $x$ by setting $x_{P_I \cup Q} := \frac{y_I}{|\mathcal{Q}_I|}$ for every $I \in \mathcal{I}(G)$ and every $Q \in \mathcal{Q}_I$.
    Note that $\sum_{P \in \paths} x_P = \sum_{I \in \mathcal{I}(G)} y_I = \ell$ by construction.
    Moreover, $x[a_{iv}] = \sum_{I \in \mathcal{I}(G) : v \in I} y_I = 1$ for all $v \in V$ and $i \in [m]$ with $e_i \in \delta_G(v)$, again by construction.

    Now consider any $S \in \mathcal{S}_{H,2}$.
    If $|S| = 1$ then $\lambda(x, S) = 1$.
    If $|S| = 2$, then $S$ must be an edges in $H$. By construction of $E_H$, there must be some $v \in V$ and $i, j \in [m]$ such that $S = \{a_{vi}, a_{vj}\}$ and $v \in e_i \cap e_j$.
    Note that our construction of $x$ implies that for every $P \in \paths$ with $x_P > 0$ it holds that $a_{vi} \in P$ if and only if $a_{vj} \in P$. 
    Hence $\lambda(x, S) = \sum_{P : P \cap S \neq \emptyset} x_P = x[a_{vi}] = x[a_{vj}] = 1$. \qed
\end{proof}

\clearpage

\section{{\mrfrx} is $P^{\NP[\log]}$-hard (Proof of \cref{thm:mrf-pnp})}
\label{app:mrfr-pnp}

In this section, we prove the following theorem, which together with \cref{thm:main-reduction} implies \cref{thm:mrf-pnp}.
\begin{theorem}\label{thm:mrfr-pnp}
    {\mrfrx} is $P^{\NP[\log]}$-hard.
\end{theorem}

To prove \cref{thm:mrfr-pnp}, we apply a theorem by \citet{wagner1987more}. 
For the formal statement of Wagner's result, we introduce the following notation.
For a decision problem $\mathcal{A}$, let $Y(\mathcal{A})$ denote the set of YES instances of $\mathcal{A}$.
We say that a sequence $I_1, \dots, I_{\ell}$ of instances of $\mathcal{A}$ is \emph{monotone} if $A_i \in \mathcal{A}$ implies $A_{j} \in Y(\mathcal{A})$ for all $j < i$.

\begin{theorem}[{\citep[Theorem~5.2]{wagner1987more}}]\label{thm:wagner}
    Let $\mathcal{A}$ be some $\NP$-complete decision problem and let $\mathcal{B}$ be some decision problem.
    If there is a polynomial-time computable function $f$ that takes as input $\ell \in \mathbb{N}$ and a monotone sequence $I_1, \dots, I_{2\ell}$ of instances of $\mathcal{A}$ and outputs an instance $B$ of $\mathcal{B}$ such that $B \in Y(\mathcal{B})$ if and only if $|\{i \in [2\ell] \st I_i \in Y(\mathcal{A})\}|$ is odd, then $\mathcal{B}$ is $P^{\NP[\log]}$-hard.
\end{theorem}

To apply \cref{thm:wagner}, we first show the following lemma.

\begin{restatable}{lemma}{restateLemCliqueFlow}\label{lem:clique-flow}
    Let $\mathcal{A}$ be some $\NP$-complete decision problem.
    There exists a poly\-nomial-time computable function $g$ that takes as input two instances $I, I'$ of $\mathcal{A}$ and outputs an instance $\bar{J}$ of {\mrfrx} such that $\bar{J} \in Y(\textup{\mrfrx})$ if and only if $I \notin Y(\mathcal{A})$ or $I' \in Y(\mathcal{A})$.
\end{restatable}

\begin{proof}
    We first transform instances $I$ and $I'$ of $\mathcal{A}$ to intermediary instances as follows.
    
    We transform $I$ into an instance $J^{\textsc{Clique}}$ of \textsc{Clique} with graph $G = (V, E)$ and integer $n \in \mathbb{N}$ such that $G$ has a clique of size $n$ if and only if $I \in Y(\mathcal{A})$. 
    Note that this transformation can be done in polynomial time since \textsc{Clique} is $\NP$-complete~\citep{garey2002computers}, and that we may additionally assume that $G$ has no clique of size strictly larger than $n$ (this is implicit in the reduction from \textsc{Satisfiability} used in \citep{garey2002computers}).

    We transform the instance $I'$ to an instance $J^{\textsc{Flow}}$ of {\mrfrx} with interdiction budget $k' = 2$, digraph $D' = (V', A')$, source $s$, sink $t$, demand $\theta'$, and compatibility graph $H' = (A', E')$, such that $J^{\textsc{Flow}}$ admits an $s$-$t$-path flow $x$ in $D'$ of value $\theta'$ with $\max_{S \in \mathcal{S}_{H',2}} \lambda(x, S) \leq k' - 1$ if and only if $I' \in Y(\mathcal{A})$.
    Note that also this transformation can be done in polynomial time due to the $\NP$-completeness of {\mrfrx} with $k = 2$ shown in \cref{thm:mrfr-2-np}. We may assume without loss of generality that the capacity of a minimum $s$-$t$-cut in $D'$ is exactly~$\theta'$ and that $H'$ is a matching (by subdividing any arcs in $A'$ that have degree larger $1$ in $H'$).

    Now let $\bar{A} := A' \cup \{a_v \st v \in V\}$, where each $a_v$ for $v \in V$ is an arc from $s$ to~$t$.
    Moreover, let $\bar{E} := E' \cup \big\{ \{a_v, a_w\} \st \{v, w\} \in E \big\} \cup \big\{ \{a_v, a'\} \st v \in V, a' \in A' \big\}$.
    We obtain the {\mrfrx} instance $\bar{J}$ using the digraph $\bar{D} := (V', \bar{A})$ with the same node set $V'$, source $s$, and sink $t$ as $J^{\textsc{Flow}}$, with interdiction budget $\bar{k} := n + 2$, demand $\bar{\theta} := \theta' + |V|$, and with compatibility graph $\bar{H} = (\bar{A}, \bar{E})$.
    See \cref{fig:reduction-pnp} for a depiction of the construction.

    Let $\bar{\paths}$ denote the set of $s$-$t$-paths in $\bar{D}$ and let $X_{\bar{\theta}}$ denote the set of $s$-$t$-flows in $\bar{D}$ with respect to unit capacities and value $\sum_{P \in \bar{\paths}} x_P = \bar{\theta}$.
    We show that there is $x \in X_{\bar{\theta}}$ with $\max_{S \in \mathcal{S}_{\bar{H},2}} \lambda(x, S) \leq \bar{k} - 1$
    if and only if $I \notin Y(\mathcal{A})$ or $I' \in Y(\mathcal{A})$.
    To this end, let $X_{\theta'}$ denote the set of $s$-$t$-flows in $D'$ with unit capacities and value $\sum_{P \in \paths'} x_P = \theta'$ (where $\paths'$ is the set of $s$-$t$-paths in $D'$).

    Note that there is a one-to-one correspondence between flows $X_{\bar{\theta}}$ and flows in $X_{\bar{\theta}}$:
    Given $x' \in X_{\theta'}$, the corresponding flow $\bar{x} \in X_{\bar{\theta}}$ is given by $\bar{x}_P = x'_P$ for any $s$-$t$-path $P$ in $D'$ and $\bar{x}_{P_v} = 1$ for all $v \in V$, where $Q_v$ is the path consisting of the single arc $a_v$.
    In particular, $X_{\bar{\theta}} \neq \emptyset$ because $X_{\theta'} \neq \emptyset$ (because the min-cut size in $D'$ is $\theta'$).
    Recalling that $H'$ is a matching, note moreover that any $\bar{S} \in \mathcal{S}_{\bar{H}, \bar{k}}$ with $|\bar{S}| > 2$ is of the form $\bar{S}(C, S') := S' \cup \{a_v \st v \in C\}$ for some $S' \in \mathcal{S}_{H', 2}$ and some clique $C$ in $G$, and that conversely $\bar{S}(C, S') \in \mathcal{S}_{\bar{H}, \bar{k}}$ for any $S' \in \mathcal{S}_{H', 2}$ and any clique $C$ of $G$ (which by our earlier assumption cannot be larger than $n = k - 2$).
    
    Finally, note that for any pair of corresponding flows $x' \in X_{\theta'}$ and $\bar{x} \in X_{\bar{\theta}}$, it holds that $\sum_{P \in \paths' : \bar{S}(C, S') \cap P \neq \emptyset} \bar{x}_P = |C| + \sum_{P \in \paths' : S' \cap P \neq \emptyset} x'_P$ for any $S' \in \mathcal{S}_{H', 2}$ and any clique $C$ in $G$.
    Thus, $\bar{J}$ is a NO instance if and only if for every $\bar{x} \in X_{\bar{\theta}}$ there is $S' \in \mathcal{S}_{H', 2}$ and a clique $C$ in $G$ with $\sum_{P \in \paths' : \bar{S}(C, S') \cap P \neq \emptyset} \bar{x}_P \geq \bar{k} \geq n + 2$, which is the case if and only if $\sum_{P \in \paths' : S' \cap P \neq \emptyset} x'_P = 2$ and $|C| = n$.
    Hence, $\bar{J}$ is a NO instance if and only if $J^{\textsc{Clique}}$ (and hence $I$) is a YES instance and $J^{\textsc{Flow}}$ (and hence $I'$) is an NO instance. \qed
\end{proof}

\begin{figure}[t]
    \centering

    \begin{tikzpicture}


        \draw[fill=black!10!white] (-0.5, 0.5) rectangle ++(3, 1);
        \path (0, 0.75) node[node, label=above:{$v_1$}] (v1) {}
            ++(1, 0) node (vdots) {\dots}
            ++(1, 0) node[node, label=above:{$v_m$}] (vn) {};
        \path (1, 1.25) node {$G$};

        \draw[fill=black!10!white] (-0.5, 0) rectangle ++(3, -1);
        \path (0, -0.25) node[node] (a1) {}
            ++(1, 0) node (adots) {\dots}
            ++(1, 0) node[node] (am) {};
        \path (1, -0.75) node {$H'$};

        \draw (v1) edge (a1);
        \draw (v1) edge (adots);
        \draw (v1) edge (am);
        \draw (vdots) edge (a1);
        \draw (vdots) edge (adots);
        \draw (vdots) edge (am);
        \draw (vn) edge (a1);
        \draw (vn) edge (adots);
        \draw (vn) edge (am);
        
        
        \path (4, -0.25) node (sc) {};
        
        \draw[fill=black!10!white] (sc) ++ (-0.2, 0.25) rectangle ++(4.4, -1);
        \path (sc) ++(2, -0.25) node {$D'$};
        
        \path (sc) node[node, label=below:{$s$}] (s) {}
            ++(4, 0)  node[node, label=below:{$t$}] (t) {};
        
        \draw[-latex, bend left=90] (s) edge node[above, sloped] {$a_{v_1}$} (t);
        \path (s) ++(2, 1.1) node {$\vdots$};
        \draw[-latex, bend left=35] (s) edge node[below, sloped] {$a_{v_m}$} (t);
        
    \end{tikzpicture}

    \vspace*{-0.2cm}
    
    \caption{Illustration of the construction from the proof of \cref{lem:clique-flow}. The figure on the left shows the compatibility graph $\bar{H}$ of the constructed instance, which consists of the graph $G$ from the clique instance $J^{\textsc{Clique}}$ with nodes $v_1, \dots, v_m$ and the compatibility graph $H'$ of the {\mrf} instance $J^{\textsc{Flow}}$ (the edges of $G$ and $H'$ are omitted in the figure). The figure on the right shows the digraph $\bar{D}$ of the constructed instance, which consists of the digraph $D'$ of $J^{\textsc{Flow}}$ (the arcs of $D'$ are omitted in the figure) and the arcs $a_v$ for each node $v$ of $G$.}
    \label{fig:reduction-pnp}

\end{figure}
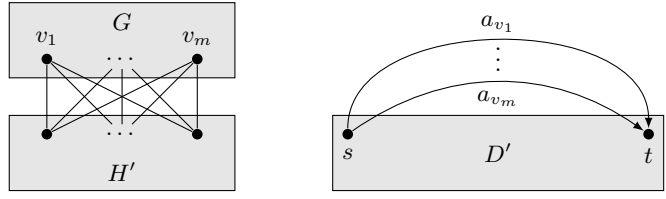

Now, to apply \cref{thm:wagner}, consider a monotone sequence $I_1, \dots, I_{2\ell}$ of instances of an arbitrary $\NP$-complete decisionn problem $\mathcal{A}$ (e.g., \textsc{Satisfiability}).
For every $r \in [\ell]$, we apply \cref{lem:clique-flow} with $I = I_{2r-1}$ and $I' = I_{2r}$ to obtain an instance $\bar{J}_r$ of {\mrfrx} with digraph $\bar{D}_r = (\bar{V}_r, \bar{A}_r)$, interdiction budget $\bar{k}_r$, demand $\bar{\theta}_r$, and compatibility graph $\bar{H}_r = (\bar{A}_r, \bar{E}_r)$ such that there is an $s$-$t$-flow $x$ of value $\bar{\theta}_r$ in $\bar{D}_r$ with $\min_{S \in \mathcal{S}_{\bar{H}_r, \bar{k}_r}} \mu(x, S)$ if and only if $I_{2i-1}$ is a NO instance or $I_{2i}$ is a YES instance for $\mathcal{A}$.

Without loss of generality, we can assume that $\bar{\theta}_r$ is the maximum flow value in $\bar{D}_r$ for unit capacities and that there is a $k \in \mathbb{N}$ such that $\bar{k}_r = k$ for all $r \in [\ell]$.
For the latter, choose $k := \max_{i \in [\ell]} \bar{k}_r$ and modify $\bar{J}_r$ for $i \in [\ell]$ with $\bar{k}_r < k$ as follows: Introduce $k - \bar{k}_r$ additional arcs from $s$ to $t$ into $\bar{D}_r$, and in the compatibility graph $\bar{H}_r$ make all of them adjacent to all other arcs; further increase the demand to $\bar{\theta}_r + k - \bar{k}_r$ and the interdiction budget to $k$. 
For the former, simply split the source $s$ into two nodes $s$ and $s'$ with a bundle of $\bar{\theta}_r$ arcs from $s$ to $s'$, all of which are singletons in $\bar{H}_r$.
The resulting instance is a YES instance if and only if the original instance was a YES instance.

We combine the instances $\bar{J}_r$ for $r \in [\ell]$ to a single instance $J$ of {\mrfrx} as follows.
To obtain the digraph $D = (V, A)$, we take the union of all digraphs $\bar{D}_r$, identifying the sources $s$ with one another and identifying the sinks $t$ with one another.
The compatibility graph $H = (A, E_H)$ consists of the union of the compatibility graphs of the instances $\bar{J}_r$.
We set the interdiction budget to $k$ (i.e., the same value as in each instance $\bar{J}_r$) and the demand to $\theta = \sum_{i \in [\ell]} \bar{\theta}_r$.

\begin{lemma}\label{lem:all-yes}
	The {\mrfrx} instance $J$ is a YES instance if and only if $\bar{J}_r$ is a YES instance for all $r \in [\ell]$.
\end{lemma}

\begin{proof}
	Note that any flow $x$ of value $\theta$ in $D$ with unit capacities is the sum of flows $\bar{x}_r$ of value $\bar{\theta}_r$ in $\bar{D}_r$ with unit capacities.
	Note further that the compatibility graph $H$ of $J$ does not contain any edges between arcs in $D$ that belong to distinct $\bar{D}_r$, and hence $S \subseteq A$ is in $\mathcal{S}_{H,k}$ if and only if there is $r \in [\ell]$ such that $S \in \mathcal{S}_{\bar{H}_r, k}$.
	From this it follows that $\max_{S \in \mathcal{S}_{H,k}} \lambda(x, S) \leq k - 1$ if and only if $\max_{S \in \mathcal{S}_{\bar{H}_r,k}} \lambda(\bar{x}_r, S) \leq k-1$ for all $r \in [\ell]$.
	Hence $J$ is a YES instance if and only if $\bar{J}_r$ is a YES instance for all $r \in [\ell]$. \qed
\end{proof}

\begin{lemma}\label{lem:exists-no}
	There exists $r \in [\ell]$ such that $\bar{J}_r$ is a NO instance if and only if $|\{i \in [2\ell] \st I_i \in Y(\mathcal{A})\}|$ is odd.
\end{lemma}

\begin{proof}
Note that $|\{i \in [2\ell] \st I_i \in Y(\mathcal{A})\}|$ is odd if and only if there exists $r \in [\ell]$ such that $I_{2r - 1} \in Y(\mathcal{A})$ and $I_{2r} \notin Y(\mathcal{A})$.
By \cref{lem:clique-flow}, $I_{2r - 1} \in Y(\mathcal{A})$ and $I_{2r} \notin Y(\mathcal{A})$ if and only if $\bar{J}_r$ is a NO instance. \qed
\end{proof}

Together, \cref{lem:all-yes,lem:exists-no} imply that $J$ is a NO instance if and only if $|\{i \in [2\ell] \st I_i \in Y(\mathcal{A})\}|$. This implies that the complement of {\mrfrx} is $P^{\NP[\log]}$-hard. Since $P^{\NP[\log]}$ is closed under complements, also {\mrfrx} is $P^{\NP[\log]}$-hard, completing the proof of \cref{thm:mrfr-pnp}.

\clearpage

\section{Integral {\mrfrx} is $\Sigma_2^P$-hard (Proof of \cref{thm:mrf-sigma})}
\label{app:mrfr-sigma}

In this section, we prove the following theorem, which together with \cref{thm:main-reduction} implies \cref{thm:mrf-sigma}.

\begin{theorem}\label{thm:integral-mrfr-sigma}
    Integral {\mrfrx} is $\Sigma_2^P$-hard.
\end{theorem}

We prove \cref{thm:integral-mrfr-sigma} by reduction from the problem \textsc{Clique Interdiction}, which is defined as follows:
Given a graph $G = (V, E)$ and two numbers $\ell, r \in \mathbb{N}$, 
the goal is to decide whether there is $R \subseteq V$ with $|R| \leq r$ such that $|C \cap R| \geq 1$ for all $C \in \mathcal{C}_{\ell}(G)$, where $\mathcal{C}_{\ell}(G)$ is the set of all cliques of size $\ell$ in $G$.
\citet{rutenburg1994propositional} showed that \textsc{Clique Interdiction} is $\Sigma_2^P$-complete.

Given an instance of \textsc{Clique Interdiction} with graph $G = (V, E)$ and numbers $\ell, r \in \mathbb{N}$, we construct an instance of \textsc{Integral {\mrfrx}} with digraph~$D = (V', A)$, interdiction budget $k := \ell + 1$, demand $\theta := |V| + \ell$, and compatibility graph $H = (A, E_H)$ described below.

First, we fix an arbitrary linear order $v_1, \dots, v_n$ of the nodes in $V$.
The node set $V'$ of $D$ is given by
    $V' := \{s, s^0, s^1, v_{n+1}, t^0, t^1, t\} \cup \{v_i, y_i^0, z_i^0, y_i^1, z_i^1 \st i \in [n]\}$.
The arc set $A$ of $D$ contains the following arcs:
\begin{itemize}
    \item the arcs $a_s := (s, v_1)$ and $a_t := (v_{n+1}, t)$,
    \item the arcs $(v_i, y^b_i)$, $(y^b_i, z^b_i)$, $(z^b_i, v_{i+1})$ for each $i \in [n]$ and $b \in \{0, 1\}$,
    \item a bundle $A^0_{s}$ of $r$ parallel arcs from $s$ to $s^0$,
    \item a bundle $A^0_{t}$ of $r$ parallel arcs from $t^0$ to $t$,
    \item arcs $(s^0, y^0_i)$ and $(z^0_i, t^0)$ for each $i \in [n]$,
    \item a bundle $A^1_{s}$ of $|V| - r$ parallel arcs from $s$ to $s^1$,
    \item a bundle $A^1_{t}$ of $|V| - r$ parallel arcs from $t^1$ to $t$,
    \item arcs $(s^1, y^1_i)$ and $(z^1_i, t^1)$ for each $i \in [n]$,
    \item a bundle $A_{st}$ of arcs $k-2$ arcs from $s$ to $t$.
\end{itemize}
For $v = v_i \in V$, we use $a^0_v$ and $a^1_v$, respectively, to denote the arcs $(y^0_i, z^0_i)$ and $(y^1_i, z^1_i)$, respectively.
We further define $\bar{S} := A_{st} \cup \{a_s, a_t\}$.
The edge set of the compatibility graph $H$ is given by 
\begin{align*}
    E_H := \big\{ \{a^1_v, a^1_w\} \st \{v, w\} \in E \big\} \cup \big\{ \{a_s, a^1_v\} \st v \in V \big\} \cup \binom{\bar{S}}{2}.
\end{align*}
See \cref{fig:reduction-mrfr-sigma} for an illustration of the constructed graph $D'$.

\begin{figure}[t]
    \centering

    \begin{tikzpicture}
        
        \path (4, 0) node[node, label=left:{$s$}] (s) {}
        
            ++(1.75, 0)  node[node, label=below:{$v_1$}] (v1) {}
            ++(0.75, 0.75)  node[node, label=below:{$y^0_1$}] (y01) {}
            +(0, -1.5)  node[node, label=above:{\;\;$y^1_1$}] (y11) {}
            ++(1.2, 0)  node[node, label=below:{$z^0_1$}] (z01) {}
            +(0, -1.5)  node[node, label=above:{$z^1_1$}] (z11) {}
            ++(0.75, -0.75)  node[node, label=below:{\;\;$v_2$}] (v2) {}
            
            ++(0.75, 0) node {\dots}
            
            ++(0.75, 0)  node[node, label=below:{$v_n$}] (vn) {}
            ++(0.75, 0.75)  node[node, label=below:{$y^0_n$}] (y0n) {}
            +(0, -1.5)  node[node, label=above:{\;\;$y^1_n$}] (y1n) {}
            ++(1.2, 0)  node[node, label=below:{$z^0_n$}] (z0n) {}
            +(0, -1.5)  node[node, label=above:{$z^1_n$}] (z1n) {}
            ++(0.75, -0.75)  node[node, label=below:{\;\;\;\;\;$v_{n+1}$}] (vnp) {}
            
            ++(1.75, 0)  node[node, label=right:{$t$}] (t) {};
        
        \draw[-latex] (s) edge node[above, pos=0.7] {$a_s$} (v1);
        \draw[-latex] (v1) edge (y01);
        \draw[-latex] (v1) edge (y11);
        \draw[-latex] (y01) edge node[above] {$a^0_{v_1}$} (z01);
        \draw[-latex] (y11) edge node[below] {$a^1_{v_1}$} (z11);
        \draw[-latex] (z01) edge (v2);
        \draw[-latex] (z11) edge (v2);
        
        \draw[-latex] (vn) edge (y0n);
        \draw[-latex] (vn) edge (y1n);
        \draw[-latex] (y0n) edge node[above] {$a^0_{v_n}$} (z0n);
        \draw[-latex] (y1n) edge node[below] {$a^1_{v_n}$} (z1n);
        \draw[-latex] (z0n) edge (vnp);
        \draw[-latex] (z1n) edge (vnp);
        \draw[-latex] (vnp) edge node[above, pos=0.3] {$a_t$} (t);

        \path (s) ++(1.5, 2) node[node, label=above:{$s^0$}] (s0) {}
            ++(0, -4)node[node, label=below:{$s^1$}] (s1) {};

        \draw[-latex, bend left=45] (s) edge (s0);
        \draw[-latex, bend right=45] (s) edge (s0);
        \draw[dotted, thick] (s) ++(25:0.5) arc (25:85:0.5);
        \path (s) ++(0.7, 1.1) node[align=center] {$A^0_s$\\[-2pt] \scriptsize ($r$ arcs)};
        
        \draw[-latex, bend left=45] (s) edge (s1);
        \draw[-latex, bend right=45] (s) edge (s1);
        \draw[dotted, thick] (s) ++(-25:0.5) arc (-25:-85:0.5);
        \path (s) ++(0.8, -1) node[align=center] {$A^1_s$\\[-2pt] \scriptsize ($n-r$ \\[-2pt] \scriptsize arcs)};
        
        \draw[-latex] (s0) edge (y01);
        \draw[dotted, thick] (s0) ++(-5:0.5) arc (-5:-45:0.5);
        \draw[-latex, out=0, in=135] (s0) edge (y0n);
        
        \draw[-latex] (s1) edge (y11);
        \draw[dotted, thick] (s1) ++(5:0.5) arc (5:45:0.5);
        \draw[-latex, out=0, in=-135] (s1) edge (y1n);

        \path (t) ++(-1.5, 2) node[node, label=above:{$t^0$}] (t0) {}
            ++(0, -4) node[node, label=below:{$t^1$}] (t1) {};
        
        \draw[-latex, bend left=45] (t0) edge (t);
        \draw[-latex, bend right=45] (t0) edge (t);
        \draw[dotted, thick] (t0) ++(-25:0.5) arc (-25:-85:0.5);
        \path (t) ++(-0.7, 1.1) node[align=center] {$A^0_t$\\[-2pt] \scriptsize ($r$ arcs)};
        
        \draw[-latex, bend left=45] (t1) edge (t);
        \draw[-latex, bend right=45] (t1) edge (t);
        \draw[dotted, thick] (t1) ++(25:0.5) arc (25:85:0.5);
        \path (t) ++(-0.8, -1) node[align=center] {$A^1_t$\\[-2pt] \scriptsize ($n-r$ \\[-2pt] \scriptsize arcs)};
        
        \draw[-latex, out=45, in=180] (z01) edge (t0);
        \draw[dotted, thick] (t0) ++(185:0.5) arc (185:225:0.5);
        \draw[-latex] (z0n) edge (t0);
        
        \draw[-latex, out=-45, in=-180] (z11) edge (t1);
        \draw[dotted, thick] (t1) ++(-185:0.5) arc (-185:-225:0.5);
        \draw[-latex] (z1n) edge (t1);
        
    \end{tikzpicture}

    \vspace*{-0.2cm}
    
    \caption{Illustration of the digraph $D$ constructed in the proof of \cref{thm:integral-mrfr-sigma}. The figure omits the bundle $A_{st}$ with $k-2$ arcs from $s$ to $t$.}
    \label{fig:reduction-mrfr-sigma}

\end{figure}
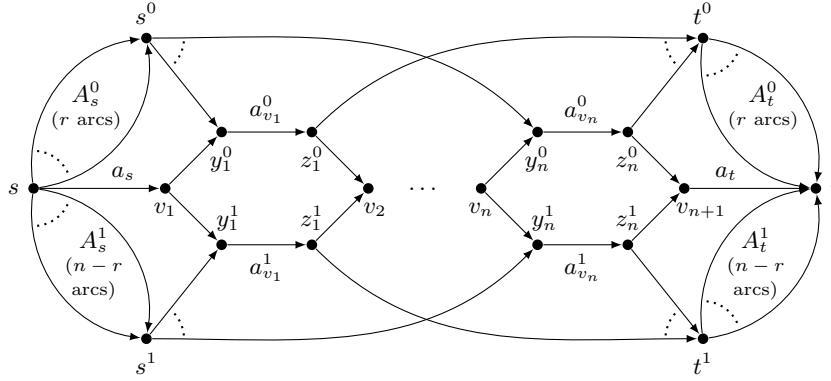

Let $\paths$ denote the set of $s$-$t$-paths in $D$ and let $X^{\star}$ denote the set of integral $s$-$t$-flows of value $\theta$ in $D$ with $\lambda(x, S) \leq k - 1$ for all $S \in \mathcal{S}_{H,k}$.
We show that $X^{\star} \neq \emptyset$ if and only if there is $R \subseteq V$ with $|R| \leq r$ and $R \cap C \neq \emptyset$ for all~\mbox{$C \in \mathcal{C}_{\ell}(G)$}.

\subsubsection*{From flows to node sets interdicting cliques.} We now show that the existence of a flow $x \in X^{\star}$ yields a set of nodes in $R$ interdicting all cliques of size~$\ell$.
We first show the existence of a flow-carrying path $P'$ that contains both~$a_s$ and~$a_t$.

\begin{lemma}\label{lem:sigma-path-s-t}
    Let $x \in X^{\star}$.
    Then there is $P' \in \paths$ with $x_{P'} = 1$ at $a_s, a_t \in P'$.
\end{lemma}

\begin{proof}
    Note that $x[a] = 1$ for all $a \in \delta^+(s) \cup \delta^-(t)$ because $x$ has value $\theta = |V| + \ell = |\delta^+(s)| = |\delta^-(t)|$.
    In particular, this implies $x[a] = 1$ for all $a \in \bar{S}$.
    Note that $\lambda(x, \bar{S}) \leq  k - 1 < \sum_{a \in \bar{S}} x[a]$ implies that there must be $P' \in \paths$ with $x_{P'} = 1$ and $|P \cap \bar{S}| \geq 2$.
    As all arcs in $A_{st}$ are starting at $s$ and ending at~$t$, the only two arcs in $\bar{S}$ that can appear simultaneously on an $s$-$t$-path are $a_s$ and~$a_t$. \qed
\end{proof}

We now show that $P'$ is complemented by a set of $|V|-r$ paths containing those arcs $a^1_v$ that are not in $P'$.

\begin{lemma}\label{lem:sigma-paths-structure}
    Let $x \in X^{\star}$ and $P'$ as in \cref{lem:sigma-path-s-t}.
    Then there are $\bar{r} := |V| - r$ paths $Q_1, \dots, Q_{\bar{r}} \in \paths$ such that $x_{Q_i} = 1$ and $|Q_i \cap \{a^1_v \st v \in V\}| = 1$ for all $i \in [\bar{r}]$.
    Moreover, $a^1_v \in P'$ for every $v \in V$ for which $a^1_v \notin Q_i$ for all $i \in [\bar{r}]$.
\end{lemma}

\begin{proof}
    Note that $\paths_x := \{P \in \paths : x_P = 1\}$ consists of $\theta = |V| + |\ell|$ arc-disjoint $s$-$t$-paths.
    These are:
    \begin{itemize}
        \item $k - 2 = \ell - 1$ paths are consisting of a single arc from $A_{st}$,
        \item $r$ paths visiting $s^0$, each of which contains at least one arc $a^0_v$ for some $v \in V$,
        \item $|V| - r$ paths visiting $s^1$, each of which contains at least one arc $a^1_v$ for some~$v \in V$,
        \item the path $P'$.
    \end{itemize}
    Because $P'$ contains both $a_s$ and $a_t$, it must visit both $v_1$ and $v_{n+1}$.
    Note that for every $v \in V$, the set $\{a^0_v, a^1_v\}$ is a $v_1$-$v_{n+1}$-cut, and hence, for each $v \in V$, the path $P'$ must contain either $a^0_v$ or $a^1_v$.
    Because $|\{a^0_v, a^1_v \st v \in V\}| = 2|V|$ and $P'$ contains $|V|$ of these nodes, each of the $r$ paths visiting $s^0$ must contain exactly one arc of the form $a^0_v$ and each of the $\bar{r} = |V| - r$ paths visiting $s^1$ must contain exactly one arc of the form $a^1_v$. 
    The latter paths thus comprise the desired paths $Q_1, \dots, Q_{\bar{r}}$. \qed
\end{proof}

\begin{lemma}\label{lem:sigma-flow-interdicts-cliques}
    Let $x \in X^{\star}$ and $P'$ as in \cref{lem:sigma-path-s-t}. 
    Let $R := \{v \in V \st a^1_v \in P'\}$.
    Then $|R| = r$ and $R \cap C \neq \emptyset$ for all $C \in \mathcal{C}_{\ell}(G)$.
\end{lemma}

\begin{proof}
    Note that for every $v \in V$, either $a^1_v \in P'$ or $a^1_v \in Q_i$ for some $i \in [\bar{r}]$, where $Q_1, \dots, Q_{\bar{r}}$ are the paths given by \cref{lem:sigma-paths-structure}.
    This implies that $|R| = r$.
    
    Now consider any $C \in \mathcal{C}_{\ell}(G)$. 
    Note that $S := \{a^1_v \st v \in C\} \cup \{a_s\} \in \mathcal{S}_{H,k}$.
    Hence, $\sum_{P : P \cap S \neq \emptyset} x_P \leq k - 1$.
    Observe that the paths intersecting $S$ are exactly the paths $P'$ and $Q_i$ for $i \in [\bar{r}]$ with $a^1_v \in Q_i$ for some $v \in C$.
    Because for every $v \in C$, either $v \in R$ or there is a unique $i \in [\bar{r}]$ with $a^1_v \in Q_i$, we obtain $\sum_{P : P \cap S \neq \emptyset} x_P = 1 + |C \setminus R|$.
    Hence $|C \setminus R| \leq k - 2 = \ell - 1$. \qed
\end{proof}

\subsubsection*{From nodes sets interdicting cliques to flows.} 
We now show that conversely, a set of $r$ nodes interdicting all cliques of size $\ell$ implies the existence of a flow $x \in X^{\star}$.

\begin{lemma}
    Let $R \subseteq V$ with $|R| \leq r$ and $C \cap R \neq \emptyset$ for all $C \in \mathcal{C}_{\ell}(G)$.
    Then $X^{\star} \neq \emptyset$.
\end{lemma}

\begin{proof}
Without loss of generality, we can assume $|R| = r$ (by adding arbitrary vertices to $R$).
We construct a flow $x \in X$ as follows:
\begin{itemize}
    \item Let $P_R$ be the unique $s$-$t$-path in $D$ containing the arcs $a_s$, $a^1_v$ for $v \in R$, $a^0_v$ for $v \in V \setminus R$, and $a_t$. Set $x_{P_R} = 1$.
    \item Fix an arbitrary order $a^0_1, \dots, a^0_{r}$ of the arcs in $A^0_s$, $\bar{a}^0_1, \dots, \bar{a}^0_{r}$ of the arcs in $A^0_t$, and $w^0_1, \dots, w^0_{|V|-r}$ of the nodes in $R$, respectively.
    Let $P^0_i$ be the path
    consisting of $a^0_i$, $a^0_{w_i}$, $\bar{a}^0_i$
    and set $x_{P^0_i} = 1$ for each $i = 1, \dots, r$.
    \item Fix an arbitrary order $a^1_1, \dots, a^1_{|V|-r}$ of the arcs in $A^1_s$, $\bar{a}^1_1, \dots, \bar{a}^1_{|V|-r}$ of the arcs in $A^1_t$, and $w^1_1, \dots, w^1_{|V|-r}$ of the nodes in $V \setminus R$, respectively.
    Let $P^1_i$ be the path
    consisting of $a^1_i$, $a^1_{w_i}$, $\bar{a}^1_i$
    and set $x_{P^1_i} = 1$ for each $i = 1, \dots, |V|-r$.
    \item Set $x_{a} = 1$ for each of the $k-2$ one-arc paths in $A_{st}$.
\end{itemize}

Note that $\sum_{P \in \mathcal{P}} x_P = |V| + 1 + k-2 = |V|+\ell=\theta$.
Now consider $S \in \mathcal{S}_k$.
We show that $\lambda(x, S) \leq k-1$.
If $|S| < k$, there is nothing to show.
For $|S| = k$, we show that $|S \cap P_r| \geq 2$, which implies $\lambda(x, S) \leq k-1$.
We distinguish the following cases:
\begin{itemize}
    \item If $S = \bar{S}$ then $|\bar{S} \cap P_r| = 2$ because $P_r$ contains both $a_s$ and $a_t$.
    \item If $S \cap \{a^1_v \st v \in V\} \neq \emptyset$ then $C := \{v \in V \st a^1_v \in S\}$ is a clique in $G$.
    If $a_s \in S$, then $|C| = k-1 = \ell$ and hence there is $v \in R \cap C$. Because $a_s \in P_r$ and $a^1_v \in P_r$, we obtain $|P_r \cap S| \geq 2$.
    If $a_s \notin S$, then $|C| = k = \ell + 1$ and hence there is $|R \cap C| \geq 2$ (as otherwise there is an $\ell$-clique contained in $C$ that is not hit by $R$). Let $v, v' \in R \cap C$ and note that $a^1_v, a^1_{v'} \in P_r \cap S$. \qed
\end{itemize}

\end{proof}

\clearpage

\section{Reduction from {\mrfrx} to {\mrfmx} (Proof of \cref{thm:mrfr-to-mrfm})}
\label{app:mrfr-to-mrfm}

In this section we provide a complete proof of \cref{thm:mrfr-to-mrfm}.
Note that the construction presented in \cref{sec:mrfr-to-mrfm} was a simplified version for the sake of conveying the high-level idea.
In order to ensure properties~\ref{prop:capacities-M} to \ref{prop:cut-M}, the actual construction presented here is more involved.
For convenience, we restate the theorem.

\restateThmMRFRtoMRFM*

Conceptually, the main differences between the construction presented below and the construction in \cref{sec:mrfr-to-mrfm} are the following:
\begin{itemize}
    \item each commodity has its own dedicated source and sink, accounting for property~\ref{prop:commodity-source-sink},
    \item the role of the commodities $A$ is taken over by commodity $0$ to accommodate property~\ref{prop:i0},
    \item every commodity from $F$ routes two units of flow through each arc to fulfill property~\ref{prop:commodity-arc},
    \item a cut of capacity $\theta$ is implanted in the initial {\mrfr} instance to ensure property~\ref{prop:cut-M}.
\end{itemize}

\subsection{Construction of the {\mrfmx} instance}
    We are given an instance of {\mrfrx} specified by a DAG $D = (V, A)$, compatibility graph $H = (A, E_H)$, interdiction budget $k \geq 2$, and demand $\theta$.
    Without loss of generality, we can assume that $|\delta^+_D(s)| = \theta$ and that there are $k-1$ parallel arcs $a_1, \dots, a_{k-1}$ from $s$ to $t$ in $D$, such that $\{a_i, a_j\} \in E_H$ for $i, j \in [k-1]$ with $i \neq j$.
    For the former property, split the source in two nodes $s$ and $\bar{s}$ connected by a bundle of $\theta$ arcs from $s$ to $\bar{s}$, each of which is a singleton in $E_H$; for the latter simply add $a_1, \dots, a_{k-1}$ to $D$, adding the clique $\{a_1, \dots, a_{k-1}\}$ to $H$, and increasing $\theta$ by $k-1$.
    
    We construct an instance of {\mrfmx} with digraph $D' = (V', A')$, capacities~$u'$, set of commodities $K$, and the same interdiction budget $k$ as as follows.
We first fix a linear order $\succ$ of the arcs of $D$ that is consistent with the DAG structure in the sense that for $a, a' \in A$ with $a \succ a'$ there is no path in $D$ from the tail of $a$ to the head of $a'$. Without loss of generality, we can assume that $a_1, \dots, a_{k-1}$ are the final $k-1$ arcs in the order (there internal order does not matter to us).
We define $F := \{(a, a') \st \{a, a'\} \notin E_H, a \succ a'\}$, i.e., $F$ contains all pairs of incompatible arcs, internally ordered by $\succ$. 
Moreover, let $F^+_a := \{(a, a') \st a' \in A, (a, a') \in F\}$ and $F^-_a := \{(a', a) \st a' \in A (a', a) \in F\}$.
Note that by construction $F^-_{a_i} = \emptyset$ for $i \in [k-1]$.
    
The maximum capacity used in our construction will be $$M := \max\{2(|A|-1)|F| + 3, \theta + 2\}.$$
The set of commodities is $K := \{0\} \cup F$.
The node set of $D'$ is defined by $$V' := V \cup \{v^+_a, v^-_a \st a \in A\} \cup \{s_i, t_i \st i \in K\}.$$
The arc set $A'$ of $D'$ and the capacities~$u'$ are given as follows. The arc set $A'$ consists of the two arcs $(s_{i_0}, s)$ and $(t, t_{i_0})$ with capacity $\theta$ each,
an arc $(v^-_a, v^+_{a'})$ with capacity $2$ for each $(a, a') \in F$,
and the following arcs for every~\mbox{$a = (v, w) \in A$}:
\begin{itemize}
    \item arc $(v, v^+_a)$ with capacity $1$,
    \item arc $(v^+_a, v^-_a)$ with capacity $M$,
    \item arc $(v^-_a, w)$ with capacity $1$,
    \item for each $i \in F \setminus F^-_a$ the arc $(s_i, v^+_a)$ with capacity $2$,
    \item for each $i \in F \setminus F^+_a$ the arc $(v^-_a, t_i)$ with capacity $2$,
    \item and arcs $(s_0, v^+_a)$ and $(v^-_a, t_0)$ both with capacity $M - 2|F| - 1$.
\end{itemize}

The demand of commodity $i \in F$ is $d_i = 2(|A|-1)$ and the demand of commodity $0$ is $d_0 = \theta + |A| \cdot (M - 2|F| - 1)$.

\subsection{Properties~\ref{prop:capacities-M} to \ref{prop:cut-M} and uniqueness of the flow of commodities $F$}
Note that the capacities in the constructed instance fulfill property~\ref{prop:capacities-M}.
One also readily verifies that the demand of each commodity $i \in K$ equals the total total out-capacity of its source $s_i$ and the total in-capacity of its sink $t_i$, i.e., $d_i = \sum_{a \in \delta^+_{D'}(s_i)} u'_a = \sum_{a \in \delta^-_{D'}(t_i)} u'_a$. 
Thus, the instance also fulfills property~\ref{prop:commodity-source-sink}.
Moreover, property~\ref{prop:i0} is fulfilled for $i_0 = 0$ (note that indeed $M = \sum_{i \in F} d_i + 3$).
Now consider the set 
\begin{align*}
    U := \{s_i \st i \in K\} \cup \{s\} \cup \{v^+_{a_i} \st i \in [k-1]\}.
\end{align*}
Note that $\delta_{D'}^+(U)$ contains the following arcs:
\begin{itemize}
    \item the $k-1$ arcs $(v^+_{a_i}, v^-_{a_i})$ having a total capacity of $(k-1)M$,
    \item for each $i \in F$, the arcs in $\delta^+_{D'}(s_i)$ except $(s_i, v^+_{a_j})$ for $j \in [k-1]$ (all of which are exit in $D'$ because $F^-_{a_j} = \emptyset$), having a total capacity of $\sum_{i \in F} (d_i - 2(k-1))$,
    \item the arcs in $\delta^+_{D'}(s_i)$ except for the arc $(s_0, s)$ and the arcs $(s_i, v^+_{a_j})$ for $j \in [k-1]$  (again all of which exist in $D'$), having a total capacity of $d_0  - \theta - (k-1) \cdot (M-2|F|-1)$,
    \item the arcs in $\delta^+_{D'}(s)$ except for the arcs $(s, v^+_{a_j})$ for $j \in [k-1]$ (all of which exit in $D'$ because $a_j$ is an out-arc of $s$ in $D$), having a total capacitiy of $\theta - (k - 1)$.
\end{itemize}
Adding up these capacities, we obtain 
\begin{align*}
    \textstyle \sum_{a \in \delta_{D'}^+(U)} u'_a \; = \;& \textstyle (k-1)M - 2(k-1)|F| + \sum_{i \in F} d_i \\
    & + d_0 - \theta - (k-1) \cdot (M-2|F|-1) + \theta - (k-1) = \textstyle \sum_{i \in K} d_i.
\end{align*}
Moreover, as each of the arcs $(v^+_{a_i}, v^-_{a_i})$ has capacity $M$, the node set $U$ fulfills the requirements of property~\ref{prop:cut-M}.
We will now show that the flow of the commodities $F = K \setminus \{i_0\}$ is actually uniquely determined by the instance. 
As a byproduct, this will also reveal that the instance fulfills property~\ref{prop:commodity-arc}.

\subsubsection{Unique flow for commodities $F$.}
Let $\paths_i$ denote the set of $s_i$-$t_i$-paths in $D'$
Let $X_K$ the set of all multicommodity flows for commodities $K$ fulfilling demand $d$ and respecting the capacities $u'$.

A useful property of the construction is that the flow of all commodities in $F = K \setminus \{0\}$ is in fact already uniquely determined by the instance, i.e., any two flows $x', x'' \in X_K$ can only differ in their flow for commodity $0$.
To see this, we define the following flow $\hat{x}$.
For each $i = (a, a') \in F$ we route the flow of commodity $i$ as follows:
\begin{itemize}
    \item route $2$ units of flow along the path $s_i$-$v^+_a$-$v^-_a$-$v^+_{a'}$-$v^-_{a'}$-$t_i$,
    \item for each $\bar{a} \in A \setminus \{a, a'\}$, route $2$ units of flow along the path $s_i$-$v^+_{\bar{a}}$-$v^-_{\bar{a}}$-$t_i$.
\end{itemize}

\begin{restatable}{lemma}{restateLemCommoditiesDetermined}\label{lem:commodities-determined}
    Let $x' \in X_K$. Then $x'_P = \hat{x}_P$ for every $i \in K \setminus \{i_0\}$ and $P \in \paths_i$.
\end{restatable}

Before we prove the lemma, observe that it indeed implies that any $x' \in X_K$ fulfills $\sum_{P \in \paths_i : (v^+_a, v^-_a) \in P} x'_P = \sum_{P \in \paths_i : (v^+_a, v^-_a) \in P} \hat{x}_P = 2$ for every $i \in K \setminus \{i_0\}$ and all $a \in A$, and that any arcs that are not of this form have capacity at most~$M - 2$. Hence the instance fulfills property~\ref{prop:commodity-arc}.

\begin{proof}[of \cref{lem:commodities-determined}]
    We first show that the aggregated arc flow of the commodities in $F_b^+$ for $b \in A$ is identical to that of $\bar{x}$.
    We then argue that this implies that that also the path flows of the individual commodities must agree.
    
    We will use the notation $x_i[a] := \sum_{P \in \paths_i : a \in P} x_P$ for any multicommodity flow $x \in D$, $i \in K$, and $a \in A$ 
    and $x_b[a] := \sum_{i \in F^+_b} x_i[a]$ for the aggregated flow of the commodities in $F^+_b$ for $b \in A$.
    We show that
    \begin{align}
        \textstyle x'_b[\bar{a}] =  \bar{x}_b[\bar{a}] \label{eq:commodities-almost-okay}
    \end{align}
    for all $\bar{a} \in A'$ and all $b \in A$.

    We start by making two useful observations.
    First, we observe that for all $i \in K$ and all $\bar{a} \in \delta_{D'}^+(s_i) \cup \delta_{D'}^-(t_i)$ it holds that  $x'_i[\bar{a}] = u'_{\bar{a}}$ and $x'_j[\bar{a}] = 0 $ for $j \in K \setminus \{i\}$, because $\sum_{a \in \delta_{D'}^+(s_i)} u'_a = d_i = \sum_{a \in \delta_{D'}^-(t_i)} u'_a$ for all $i \in K$. 
    In particular, \eqref{eq:commodities-almost-okay} holds for all $\bar{a} \in A'$ that are incident to a source or sink of $D'$.
    Second, we observe that $(v^+_a, v^-_a)$ is the only arc leaving $v^+_a$ and the only arc entering $v^-_a$ and therefore flow conservation dictates
    $\sum_{\bar{a} \in \delta_{D'}^-(v^+_a)} x'_i[\bar{a}] = x'_i[(v^+_a, v^-_a)] = \sum_{\bar{a} \in \delta_{D'}^+(v^-_a)} x'_i[\bar{a}]$ 
    for all $i \in K$ and all $a \in A$.
    Combining these two observations, we obtain
    \begin{align}
        x'_i[(v^+_a, v^-_a)] & \textstyle = u'_{(s_i, v^+_a)} + x'_i[(v, v^+_a)] + \sum_{a' : (a', a) \in F^-_a} x'_i[(v^-_{a'}, v^+_a)] \notag \\
        & \textstyle  = u'_{(v^-_a, t_i)} + x'_i[(v^-_a, w)] + \sum_{a' : (a, a') \in F^+_a} x'_i[(v^-_a, v^+_{a'})] \label{eq:multiflow-bounds}
    \end{align}
    for all $a = (v, w) \in A$ and $i \in K$, where we define $u'_{(s_i, v^+_a)} := 0$ for $i \in F^-_a$ and $u'_{(v^-_a, t_i)} := 0$ for $i \in F^+_a$ to reflect the fact that the respective arcs do not exist in $D'$.
    
    To show \eqref{eq:commodities-almost-okay}, we use the fact that every arc in $\bar{a} \in A'$ is incident to at least one node $v^+_a$ or $v^-_a$ for some $a \in A$.
    By induction over the position of $a \in A$ in the linear order of $\succ$, starting from its maximal element, we show that \eqref{eq:commodities-almost-okay} holds for all arcs $\bar{a}$ incident to $v^+_a$ or $v^-_a$ and for all $b \in A$.
    
    For the base case, let $a$ be the maximal element of $\succ$.
    Note that $\succ$-maximality of $a$ implies $a = (v, t)$ for some $v \in V$ and that $F^-_a = \emptyset$.
    The incoming arcs of $v^+_a$ in $D'$ are $(v, v^+_a)$ and the arcs $(s_i, v^+_a)$ for $i \in K$.
    The outgoing arcs of $v^-_a$ in $D'$ are $(v^-_a, t)$, the arcs $(v^-_a, t_i)$ for $i \in K$ and the arcs $(v^-_a, v^+_{a'})$ for $(a, a') \in F$.
    Combining \eqref{eq:multiflow-bounds} with $x'_i[(v^-_a, t)] = 0$ and $F^-_a = \emptyset$ yields
    \begin{align*}
       2 + x'_i[(v, v^+_a)] 
        = \elsum{a' : (a, a') \in F^+_a} x'_i[(v^-_a, v^+_{a'})]
    \end{align*}
    for all $i \in F^+_a$,
    which is only possible if $x'_i[(v, v^+_a)] = 0$ for all $i \in 
    F_a^+$ and $\sum_{i \in F^+_a} x'_i[(v^-_a, v^+_{a'})] = 2$ for all $(a, a') \in F^+_a$, because $u_{(a, a')} = 2$.
    We conclude 
    \begin{align*}
        \textstyle x'_b[(v^-_a, v^+_{a'})] = 2 = \bar{x}_b[(v^-_a, v^+_{a'})]
    \end{align*}
    for $b = a$ and all $(a, a') \in F^+_a$.
    For $i \in F \setminus F^+_a$, this further implies $x'_i[(v^-_a, v^+_{a'})] = 0$ for all $(a, a') \in F^+_a$ and thus $2 + x'_i[(v, v^+_a)] = 2$ by \eqref{eq:multiflow-bounds}.
    Thus, $x'_i[(v, v^+_a)] = 0$ for all $i \in K \setminus \{i_0\}$, concluding the base case.

    We now carry out the induction step for $a = (v, w) \in A$.
    Because $a \prec a'$ for all $a' \in \delta^+_{D}(w)$ and because $\delta^+_{D'}(w) = \{(w, s^+_{a'}) \st a' \in \delta^+_{D}(w)\}$, 
    the induction hypothesis implies $x_i[\bar{a}] = 0$ for all $\bar{a} \in \delta_{D'}^+(w)$ and all $i \in K \setminus \{i_0\}$.
    Thus, by flow conservation, $x_i[(v^-_a, w)] = 0$ for all $i \in K \setminus \{i_0\}$.

    Now let $b \in A$. We distinguish three cases.
    First, for $b = a$, summing \eqref{eq:multiflow-bounds} over all $i \in F^+_{a}$ 
    and using the fact that $x'_i[(v^-_{a'}, v^+_{a})] = 0$ for all $i \in F^+_a$ and all $(a', a) \in F$ by induction hypothesis, we obtain
        \begin{align*} 
            \textstyle 2 |F^+_{a}| + x'_a[(v, v_a^+)] = \sum_{a' : (a, a') \in F} x'_a[(v^-_a, v^+_{a'})],
        \end{align*}
    implying $x'_i[(v, v^+_a)] = 0$ for all $i \in F^+_a$ and $x'_a[(v^-_a, v^+_{a'})] = 2$ for all $(a, a') \in F$.
    Note that this in particular implies $x'_b(v^-_a, v^+_{a'})] = 0$ for all $(a, a') \in F$ and $b \in A \setminus \{a\}$ because $u'_{(v^-_a, v^+_{a'})} = 2$.
    
    Second, if $(b, a) \in F^-_a$, then 
    summing \eqref{eq:multiflow-bounds} over all $i \in F^+_{b}$ 
    and using the fact that $x'_b[(v^-_{b}, v^+_{a})] = 2$ as well as $x_b[(v^-_{a'}, v^+_{a})] = 0$ for all $(a', a) \in F^-_a$ with $a' \neq b$  yields
        \begin{align*} 
            \textstyle 2 |F^+_{b} \setminus \{(b, a)\}| + x'_b[(v, v_a^+)] + 2 = 2|F^+_{b}| + \sum_{a' \in F^+_a} x'_i[(v^-_a, v^+_{a'})] = 2|F^+_{b}|,
        \end{align*}
    implying $x'_b[(v, v^+_{a'})] = 0$.

    Third, for $b \in A \setminus \{a\}$ with $(b, a) \notin F^-_a$, summing \eqref{eq:multiflow-bounds} over all $i \in F^+_{b}$ and using the fact that $x_b[(v^-_{a'}, v^+_{a})] = 0$ for all $(a', a) \in F^-_a$ yields
        \begin{align*} 
            2 |F^+_{b}| + x'_b[(v, v_a^+)] = 2 |F^+_{b}| + \sum_{a' \in F^+_a} x'_i[(v^-_a, v^+_{a'})] = 2 |F^+_{b}|
        \end{align*}
    implying $x'_b[(v, v^+_{a'})] = 0$.
    This completes the induction.

    It remains to show that \eqref{eq:commodities-almost-okay} implies the lemma, i.e., $x'$ and $\bar{x}$ restricted to commodities $K \setminus \{i_0\}$ are not only identical with respect to their aggregated arc flow values but also the path flows of the individual commodities.
    For this, consider any commodity $i = (b, a') \in F^+_b$ for some $b \in A$.
    Note that we have already shown that $x_i[\bar{a}] = 0$ for all arcs $\bar{a}$ that are not of the form $(s_i, v^+_a)$, $(v^-_a, t_i)$, $(v^+_a, v^-_a)$ for some $a \in A$ or $(v^-_b, v^+_a)$ with $(b, a) \in F^+_b$.
    We will show that in fact $x_i[(v^-_b, v^+_{a'})] = 2$ and $x_i[(v^-_b, v^+_a)] = 0$ for $a \neq a'$. 
    Then the arc flow of commodity $i$ only supports a unique path decomposition that is identical with that of $\bar{x}$.
    To show that $x_i[(v^-_b, v^+_a)] = 2$, consider $U := 
    V' \setminus \{t_i, v^+_a, v^-_a\}$ and let $C := \{(v^-_{a'}, t_i) \st a' \in A \setminus \{b, a\}\} \cup \{(v^-_b, v^+_a)\}$.
    Note that $\delta^+_{D'}(U)$ is an $s_i$-$t_i$-cut and that $x'_i[\bar{a}] = 0$ for all $\bar{a} \in \delta^+_{D'}(U_i) \setminus C$. Since $u'[\bar{a}] = 2$ for all $\bar{a} \in C$ and $d_i = 2(|A| - 1) = 2|C|$, we conclude that all $x_i[\bar{a}] = 2$ for all $\bar{a} \in C$.
    As, by the same argument, we obtain $x_j[(v^-_b, v^+_{a'})] = 2$ for $j = (b, a') \in F^+_b$, we conclude that $x_i[(v^-_b, v^+_{a'})] = 0$ for all $a' \neq a$. \qed
\end{proof}

\subsection{Transformation of solutions}

    Let $\paths$ denote the set of $s$-$t$-paths in $D$ and let $X^{\star}$ denote the set of integral $s$-$t$-flows of value $\theta$ in $D$ with $\lambda(x, S) \leq k - 1$ for all $S \in \mathcal{S}_{H,k}$.
    Let $\mathcal{S}'_k := \{S \subseteq A' \st |S| \leq k\}$.
    We show that there is $x \in X^{\star}$ if and only if $x' \in X_K$ with $\lambda(x, S) \leq kM - 1$ for all $S \in \mathcal{S}'_k$, which completes the proof of \cref{thm:mrfm-to-mrf}.

    \begin{lemma}
        If $X^{\star} \neq \emptyset$, then there is $x' \in X_K$ with $\lambda(x, S) \leq kM - 1$ for all $S \in \mathcal{S}'_k$.
    \end{lemma}

    \begin{proof}
    Let $x \in X^{\star}$.
    We construct a flow $x' \in X_K$ as follows.
    For every $i \in K \setminus \{i_0\}$ and every $P \in \paths_i$, set $x'_P := \bar{x}_P$.
    For every $P \in \paths$, let $Q_P$ be the $s$-$t$-path in $D'$ consisting of the arcs in $\{(v, v^+_a), (v^+_a, v^-_a), (v^-_a, w) \st a = (v, w) \in P\}$ and set $x'_{Q_P} := x_P$.
    Furthermore, for every $a \in A$, let $Q_a$ be the path consisting of the arcs $(s, v^+_a)$, $(v^+_a, v^-_a)$, $(v^-_a, t)$, and set $x'_{Q_a} := M-2|F|-1$.

    Note that $\sum_{P \in \paths_i} x'_P = \sum_{P \in \paths_i} \bar{x}_P = d_i$ for $i \in K \setminus \{i_0\}$.
    Moreover, observe that $\sum_{P \in \paths_{i_0}} x'_P = \sum_{P \in \paths} x_P + \sum_{a \in A} (M-2|F|-1) = d_{i_0}$.
    It is also easy to verify that $x'[a] \leq u'_{a}$ for all $a \in A'$.
    Thus $x' \in X_K$.
    
    Now assume by contradiction there is $S' \in \mathcal{S}'_k$ with $\lambda (x', S') > kM - 1$.
    Note that this is only possible if $u'_{\bar{a}} = M$ for all $\bar{a} \in S'$, i.e., $S' = \{(v^+_a, v^-_a) \st a \in S\}$ for some $S \in \mathcal{S}_k$.
    If $S \in \mathcal{S}_{H, k}$, then $\lambda(x, S) \leq k - 1$.
    Hence, 
    \begin{align*}
        \lambda(x', S') & \textstyle \leq \sum_{a \in S'} (\bar{x}[a] + x'_{Q_a}) + \lambda(x, S)\\
        & = |S|(2|F| + M - 2|F| - 1) + \lambda(x, S)\\
        & \leq |S|(M - 1) + k - 1 \leq kM - 1.
    \end{align*}
    Furthermore, if $S \notin \mathcal{S}_{H, k}$, then there is $a, a' \in S$ with $a \neq a'$ and $\{a, a'\} \notin E_H$.
    Assume w.l.o.g., that $a \succ a'$ and let $Q$ be the path $s_i$-$v^+_a$-$v^-_a$-$v^+_{a'}$-$v^+_{a'}$-$t_i$ for commodity $i = (a, a') \in F$. Then $|Q \cap S'| = 2$, which implies
    \begin{align*}
        \lambda(x', S') & \textstyle \leq \sum_{a \in S'} u'_{a'} - \bar{x}[Q] = |S'|M - 2 \leq kM - 2.
    \end{align*}
    We conclude that there is no $S' \in \mathcal{S}'_k$ with $\lambda (x', S') > kM - 1$. \qed
\end{proof}

    \begin{lemma}
        If there is $x' \in X_K$ with $\lambda(x, S) \leq kM - 1$ for all $S \in \mathcal{S}'_k$, then $X^{\star} \neq \emptyset$.
    \end{lemma}

\begin{proof}

    \cref{lem:commodities-determined} implies that $x'_P = \bar{x}_P$ for all $i \in K \setminus \{i_0\}$ and all $P \in \paths_i$.
    We further establish that we can also assume that the flow of commodity $i_0$ in $x'$ adheres to a certain structure.
    We call a path $P \in \paths'_{i_0}$ \emph{invalid}, if there is $a = (v, w) \in A$ such that $|\{(v, v^+_a), (v^-_a, w)\} \cap P| = 1$.
    We show that without loss of generality, we can assume $x'_P = 0$ for all invalid $P \in \paths'_{i_0}$.

    Consider any $a = (v, w) \in A$.
    Recall that $x'_i[(s_0, v^+_a)] = M - |2F| - 1 = x'_i[(v^-_a, t_0)] $, because the in-capacities of $s_0$ and out-capacities of $t_0$ both equal the demand $d_0$.
    Moreover, $x'_i[(v^-_{b}, v^+_{b'})] = 0$ for all $(b, b') \in F$.
    Thus, we obtain $x'_{i_0}[(v, v^+_a)] = x'_{i_0}[(v^-_a, w)]$ by flow conservation.
    Hence if there is $Q \in \paths_{0}$ with $x'_Q > 0$ and $(v, v^+_a) \in Q$ but $(v^-_a, w) \notin Q$, then there is also $R \in \paths_{0}$ with $x'_Q > 0$ and $(v^-_a, w) \in R$ but $(v, v^+_a) \notin R$.
    Note that this also implies $(v^-_a, t_{0}) \in Q$ and $(s_{0}, v^+_a) \in R$.
    Let $Q'$ be the path consisting of the concatenation of the prefix $Q[s_{0}, v^+_a]$, the arc $(v^+_a, v^-_a)$, and the suffix $R[v^-_a, t_{0}]$.
    Let $R'$ be the path consisting of the arcs $(s_{0}, v^+_a)$, $(v^+_a, v^-_a)$, $(v^-_a, t_{0})$.
    Let $x''$ be the flow constructed from $x'$ by reducing the flow on both $Q$ and $R$ by $\varepsilon := \min \{x'_Q, x'_R\}$ and increasing the flow on both $Q'$ and $R'$ by $\varepsilon$.
    Note that $x'_i[\bar{a}] = x''_i[\bar{a}]$ for all $\bar{a} \in A'$ by construction.
    Now consider any $S \in \mathcal{S}'_k$.
    We show that $\lambda(x'', S) \leq kM - 1$.
    Note that this is trivially the case if there is $\bar{a} \in S$ with $u'_{\bar{a}} < M$.
    Hence we can assume that $S \subseteq \{(v^+_{a'}, v^-_{a'}) \st a' \in A\}$.
    We distinguish three cases.
    \begin{itemize}
        \item $Q' \cap S = \emptyset = R' \cap S$: In this case $\lambda(x'', S) \leq \lambda(x', S) \leq kM - 1$ because $x''_P = x'_P$ for all $P \in \paths' \setminus \{Q', R'\}$.
        \item $Q' \cap S \neq \emptyset$, $R' \cap S = \emptyset$: Then $R \cap S \neq \emptyset$ or $Q \cap S \neq \emptyset$ because $Q' \subseteq Q \cup R$. Thus $\lambda(x'', S) = \lambda(x'_P, S) \leq kM - 1$.
        \item $R' \cap S \neq \emptyset$: Note this implies $(v^+_a, v^-_a) \in S$ and hence $Q, R, Q', R'$ are all intersected by $S$. Therefore $\lambda(x'', S) = \lambda(x', S) \leq kM - 1$.
    \end{itemize}
    We conclude that $x'' \in X_K$ and $\lambda(x'', S) \leq kM - 1$ for all $S \in \mathcal{S}'_k$.
    Therefore, we can assume without loss of generality that $x'_P = 0$ for all invalid paths $P \in \paths_0$.

    We call a path $P \in \paths_0$ \emph{long} if starts with $(s_0, s)$ and ends with $(t_0, t)$.
    Note that for every long $P \in \paths_{0}$, the arcs $\{a \in A \st (v^+_a, v^-_a) \in P\}$ form an $s$-$t$-path in $D$, which we denote by $\pi(P)$. Note that $\pi$ is in fact a bijection between the $s$-$t$-paths of $D'$ and the corresponding $s_0$-$t_0$-paths using the subdivided arcs.
    
    Define flow $x$ in $D$ by $x_P := x'_{\pi^{-1}(P)}$ for $P \in \paths$.
    Note that $x[a] \leq  1$ for all $a = (v, w) \in A$, because \begin{align*}
        \textstyle \sum_{P \in \paths : a \in P} x_P \leq \sum_{P \in \paths_0 : (v, v^+_a) \in P} x'_P \leq u'_{(v, v^+_a)} = 1.
    \end{align*}
    Moreover, because $x'$ does not have any invalid paths, every path $P$ starting with $(s, s_0)$ and $x'_P > 0$ is long, and hence $\sum_{P \in \paths} x_P = x'_0[(s, s_0)] = \theta$.
    Let $S \in \mathcal{S}_{H,k}$ and $\bar{S} := \{(v^+_a, v^-_a) : a \in S\}$.
    Note that because $S$ is a clique in $H$, there is no $a, a' \in S$  with $(a, a') \in F$. Therefore $|P \cap \bar{S}| \leq 1$ for all $P \in \paths_{i}$ with $x'_P > 0$ for $i \in K \setminus \{0\}$.
    Similarly, $|P \cap \bar{S}| \leq 1$ for every valid paths $P \in \paths_{0}$ that are not long, because such a $P$ contains exactly one arc of the form $(v^+_a, v^-_a)$.
    Therefore
    \begin{align*}
        kM - 1 & \geq \textstyle  \lambda(x', \bar{S}) = k(M - 1) + \sum_{P \in \paths^{\star}_0 : P \cap P \cap \bar{S} \neq \emptyset} x'_P \\
        & \textstyle = k(M - 1) + \sum_{P \in \paths : P \cap S \neq \emptyset} x_P = k(M-1) + \lambda(x, S),
    \end{align*}
    where $\paths^{\star}_0$ denotes the long paths in $\paths_0$.
    We conclude that $\lambda(x, S) \leq k - 1$. \qed
\end{proof}

\clearpage

\section{Reduction from {\mrfmx} to {\mrfdec} (Missing proofs and details from \cref{sec:mrfm-to-mrf})}
\label{app:mrfm-to-mrf}

\subsection{Immune arcs}
\label{app:mrfm-to-mrf-immune}

In the construction presented in \cref{sec:mrfm-to-mrf}, we make use of immune arcs that cannot be interdicted.
While such arcs are not part of the problem definition, their use is without loss of generality in this context, and they are only used here to simplify the presentation of the construction.
To simulate an immune arcs $a = (v, w)$ with capacity $u_a \in \mathbb{Z}_+$, one simply replaces it by a bundle of $u_a$ regular (i.e., not immune) arcs from $v$ to $w$, each with capacity $1$.
As we show in the analysis of our construction, a YES instance of {\mrfm} results in an instance of {\mrfdec} in which there is a flow $x'$ of value $\Delta$ with $\lambda(x', S) \leq kM - 1$ (when forbidding to interdict the immune arcs).
In the instance where the immune arcs are replaced by bundles of capacity-$1$ arcs, one can derive a flow $x''$ from $x'$ by allocating flow paths using an immune arc arbitrarily to the arcs in the bundle (when flow is integral, this allocation can also be done integrally).
Note that for any $S \subseteq A'$ containing at least one arc of capacity $1$ it holds that $\lambda(x', S) \leq (k-1)M' + 1 < kM - 1$, and hence the instance with the bundles replacing the immune arcs is still a YES instance. Conversely, introducing bundles can not turn a NO instance into a YES instance as it gives the adversary more power.

   \subsection{Constructing a single-commodity flow from a multicommodity flow (Proof of \cref{lem:multi-to-single})}

    \begin{remark}[Preserving integrality of the flow]
    We remark that the flow $\psi(x)$ constructed from $x \in X_K$ as described in \cref{sec:mrfm-to-mrf} might not be integral even when $x$ is integral.
    However, if $x$ is integral, one can use an alternate construction where one first virtually splits each $s_i$-$t_i$-path with $x_P = \ell$ into $\ell$ paths of value one and then picks a perfect matching between the $d_i$ paths of value one and the arcs in $A_i$, extending each path using the arc it is matched to.
    \end{remark}

\restateLemMultiToSingle*
\begin{proof}
    Note that $\sum_{P \in \paths'} x'_P = \Delta$.
    We show that $\lambda_k(x') \leq kM - 1$. To this end, consider any $S \subseteq A'$ of at most $k$ non-immune arcs.
    We distinguish four cases and show that $\lambda(x', S) \leq kM - 1$ in each case.
    \begin{itemize}
        \item Case 1: $\max_{a \in S} u'_a \leq M$ and $u'_a < M$ for at least one $a \in S$.
        In this case, $\lambda(x', S) \leq \sum_{a \in S} u'_a \leq kM - 1$.

        \item Case 2: $u'_a = M$ for all $a \in S$. 
        This implies that $S \subseteq A$, i.e., all arcs in $S$ are contained in the digraph $D$ of the {\mrfmx} instance.
        Then $\lambda_k(x) \leq kM - 1$ implies $\lambda(x', S) = \lambda(\psi(x), S) = \lambda(x, S) \leq kM - 1$.
        
        \item Case 3: There is $a' \in S$ with $u'_{a'} > M$. 
        Note that the only non-immune arcs in $D'$ with capacity larger than $M$ are the arcs $(s'_i, s_i)$ for $i \in K'$, which have capacity $M' = M+2k-3$. 
        Hence let $i \in K'$ be such that $a' = (s'_i, s_i)$.
        We show that 
        \begin{align}
            \textstyle \zeta_{a} := x'[a] - \sum_{P \in \paths' : a, a' \in P} x'_P \leq M - 2 \label{eq:M-2-flow}
        \end{align} 
        for all non-immune arcs $a \in A'$.
        Note that this implies $$\textstyle \lambda(x', S) \leq x'[a'] + \sum_{a \in S \setminus \{a'\}} \zeta_a \leq M' + (k-1) (M-2) = kM - 1.$$
        To see that \eqref{eq:M-2-flow} holds for all $a \in A$, consider the following cases:
        \begin{itemize}
            \item Case 3a: $a = (s'_{j}, s_{j})$ for some $j \in K'$. This implies that both $a$ and $a'$ are contained in all paths of the form $s'$-$s'_1$-$s_1$-$\dots$-$s'_m$-$s_m$-$\bar{s}$-$t'$ used in the construction of $\bar{x}$). Therefore $\sum_{P : a, a' \in P} x'_{P} \geq 2k$ and hence $\zeta_{a} \leq M' - 2k = M - 3$.
            \item Case 3b: $a \in A$, i.e., $a$ is part of the digraph $D$. If $u_a \leq M - 2$, then $\zeta_a \leq M - 2$. If $u_a > M - 1$, then property~\ref{prop:commodity-arc} of \cref{thm:mrfr-to-mrfm} guarantees
            $\sum_{Q \in \paths' : a, a' \in Q} x'_Q = \sum_{P \in \paths_i : a \in P} x_P \geq 2$, and hence again $\zeta_{a} \leq M - 2$.
            \item Case 3c: $a \in \bar{A}$. Then there is a path $P$ of the form $s'$-$s'_1$-$s_1$-$\dots$-$s'_m$-$s_m$-$\bar{s}$-$t'$ containing both $a$ and $a'$ with $\bar{x}_P = 1$, implying $\zeta_{a} \leq u'_a - 1 = M - 2$.
            \item Case 3d: $a \in A_i$ for some $j \in K'$.
            If $i = j$ then $\sum_{P \in \paths': a, a'} x'_P \geq \psi(x)[a] = 1$ because all flow-carrying paths of $\psi(x)$ using $a \in A_i$ are extensions of $s_i$-$t_i$-paths and thus also use $a' = (s'_i, s_i)$.
            If $i \neq j$, then $\sum_{P \in \paths': a, a'} x'_P \geq \bar{x}_{P_{ija}} = 1$, where $P_{ija}$ is the path of the form $s'$-$s'_i$-$s_i$-$t_j$-$t'$ containing $a \in A_j$ as its last arc.
            Thus, in both cases $\sum_{P \in \paths': a, a' \in P} x'_P \geq 1$ and therefore $\zeta_a \leq u'_a - 1 = M - 2$.
            \item Case 3e: $u'_a \leq M - 2$. In this case $\zeta_a \leq x'[a] \leq M - 2$.
        \end{itemize}
    \end{itemize}
    We thus conclude $\lambda(x', S) \leq kM - 1$ in all cases as desired, and therefore $\min_{S \in \mathcal{S}'_k} \sum_{P \in \paths' : P \cap S = \emptyset} \geq \Delta - (kM - 1)$. \qed
\end{proof}

\subsection{Constructing a multicommodity flow from a single-commodity flow (Proof of \cref{lem:hat-x})}

To prove \cref{lem:hat-x}, we first show the following intermediary result. 

\begin{lemma}\label{lem:sauration}
    Let $x' \in X'$ with $\min_{S \in \mathcal{S}'_k} \sum_{P \in \paths' : P \cap S = \emptyset} x'_P \geq \Delta - (kM - 1)$.
    Then $\sum_{P \in \paths'} x'_P = \Delta$ and $x'$ sends $2k$ units of flow along paths of the form $s'$-$s'_1$-$s_1$-$s'_2$-$s_2$-$\dots$-$s'_m$-$s_m$-$\bar{s}$-$t'$.
\end{lemma}

\begin{proof}
    By property~\ref{prop:cut-M} of \cref{thm:mrfr-to-mrfm}, there exists $U \subseteq V$ such that $s_0, \dots, s_m \in U$ and $t_0, \dots, t_m \notin U$ and such that $\sum_{a \in \delta^+(U)} u_a = \sum_{i \in K} d_i$ and $\delta_{D}^+(U)$ contains a set $S_U$ of $k - 1$ arcs of capacity $M$.
    Let $$U' := U \cup \{s', \bar{s}\} \cup \{s'_i \st i \in K\}.$$ and consider the corresponding $s'$-$t'$-cut $C := \delta_{D'}^+(U')$. 
    Note that 
    \begin{align*}
        \textstyle \sum_{a \in C} u'_a & \textstyle = \sum_{a \in \bar{A}} u'_a + \sum_{i \in K} \left(u'_{(s',t_i)} + u'_{(s_i, t')} + \sum_{j \in K' \setminus \{i\}} u'_{(s_i, t_j)}\right) + \sum_{a \in \delta_{D}^+(U)} u_a\\
        & = \textstyle 2k(M-1) + \sum_{i \in K'} \left(\xi_i^s + \xi_i^t + (m - 1) d_i \right) + \sum_{i \in K} d_i\\\
        & =  \Delta.
    \end{align*}
    Note that $\bar{A} \subseteq C$, and hence $\sum_{a \in C \setminus (S_U \cup \{\bar{a}\})} u'_a = \Delta - (kM - 1)$ for any \mbox{$\bar{a} \in \bar{A}$} because $u'_{a} = M$ for all $a \in S_U$ and $u'_{\bar{a}} = M-1$ for all $\bar{a} \in \bar{A}$. 
    From  $S_U \cup \{\bar{a}\} \in \mathcal{S}'_k$ we obtain $\sum_{P \in \paths' : P \cap (S_U \cup \{\bar{a}\}) = \emptyset} x'_P \geq \Delta - (kM-1)$ for any $\bar{a} \in \bar{A}$, which by the above is only possible if $x'[a] = u'_a$ for all $a \in C \setminus S_U$.
    In particular, $x'[a] = M-1$ for all $a \in \bar{A}$, which is only possible if $x[(s_m, \bar{s})] = 2k$.
    
    Let $U_1 := \{s', \bar{s}, s'_1\}$ and observe that
    \begin{align}
        \textstyle \sum_{P \in \paths'} x'_P & \textstyle  \leq \sum_{a \in \delta^+_{D'}(U_1)} x'[a] - \sum_{a \in \delta^-_{D'}(U_1)} x'[a] \label{eq:cut-s1}\\
        & \textstyle \leq \sum_{a \in \delta^+_{D'}(U_1) \setminus \{(s'_1, s_1)\}} u'_a + x'[(s'_1, s_1)] - x'[(s_m, \bar{s})] \notag\\
        & = \Delta + 2k - M' + x'[(s'_1, s_1)] - 2k \notag\\ 
        & = \Delta + x'[(s'_1, s_1)] - M'. \notag
    \end{align}
    For $a \in \bar{A}$, let $\varrho_{a} := \sum_{P \in \paths' : (s'_1, s_1), a \in P} x'_P$.
    Note that $\sum_{a \in \bar{A}} \varrho_{a} \leq 2k$,  because every path containing both $(s'_1, s_1)$ and some $a \in \bar{A}$ must traverse $(s_m, \bar{s})$.
    Thus, by averaging, there is $\bar{S} \in \binom{\bar{A}}{k-1}$ with $\sum_{a \in \bar{S}} \varrho_{a} \leq 2k \cdot (k-1) / |\bar{A}| = k-1$, which yields
    \begin{align}
        \lambda(x', \bar{S} \cup \{(s'_1, s_1)\}) &  \textstyle = x[(s'_1, s_1)] + (k-1)(M-1) - \sum_{a \in \bar{S}} \varrho_{a} \label{eq:barA}\\
        & \geq x[(s'_1, s_1)] + (k-1)(M - 2). \notag
    \end{align}
    Combining \eqref{eq:cut-s1} and \eqref{eq:barA}, we obtain 
    \begin{align*}
        \Delta - (kM - 1) & \textstyle \leq \sum_{P \in \paths'} x'_P - \lambda(x', \bar{S} \cup \{(s'_1, s_1)\})\\
        & \leq \Delta + x'[(s'_1, s_1)] - M' - (x[(s'_1, s_1)] + (k-1)(M-2)) \\
        & = \Delta - (kM - 1),
    \end{align*}
    from which we conclude that both \eqref{eq:cut-s1} and \eqref{eq:barA} must hold with equality throughout.
    This in particular implies that $x'[a] = u'[a]$ for all $a \in \delta_{D'}^+(U_1) \setminus \{(s'_1, s_1)\}$ and that $x'$ sends one unit of flow along each of the $2k$ paths with the node sequence $s'$-$s'_1$-$s_1$-$s'_2$-$s_2$-$\dots$-$s'_m$-$s_m$-$\bar{s}$-$t'$ (one for each $a \in \bar{A}$).

    Now let $\bar{S} \in \binom{\bar{A}}{k-1}$ and note that 
    \begin{align*}
        \lambda(x', \bar{S} \cup \{(s'_2, s_2)\}) = M' + (k-1)(M-2) = kM - 1,
    \end{align*} where the first identity follows from the fact that $(s'_2, s_2)$ and any $\bar{a} \in \bar{A}$ share exactly $1$ unit of flow on a common path (the only path containing both these arcs is the one with the node sequence $s'$-$s'_1$-$s_1$-$s'_2$-$s_2$-$\dots$-$s'_m$-$s_m$-$\bar{s}$-$t'$).
    Now observe that $\sum_{P \in \paths' : P \cap \bar{S} \cup \{(s'_2, s_2)\} = \emptyset} x'_P \geq \Delta - (kM -1)$ and hence $\sum_{P \in \paths'} x'_P = \Delta$. \qed
\end{proof}

\restateLemHatX*

\begin{proof}

    From \cref{lem:sauration} we know that $x'$ sends $\Delta$ units of flow from $s'$ to $t'$, of which $2k$ units traverse paths of the form $s'$-$s'_1$-$s_1$-$s'_2$-$s_2$-$\dots$-$s'_m$-$s_m$-$\bar{s}$-$t'$.
    In particular, this implies that all arcs in $A' \setminus A$ are saturated by $x'$, including all arcs incident to $s'$ and $t'$.
    
    For $i \in K'$ and $a \in A_i$, let
    $\gamma_{ai} := \sum_{P \in \paths' : (s'_i, s_i), a \in P} x'_P$.
    We show that $\gamma_{ai} \geq 1$.
    For this, consider $S := \{(s'_i, s_i), a\} \cup \bar{S}$, where $\bar{S} \in \binom{\bar{A}}{k-2}$ is chosen arbitrarily.
    Note that 
    \begin{align*}
        \lambda(x', S) & = M' + (k -1) (M-1) - |\bar{S}| - \gamma_{ai} = kM - \gamma_{ai},
    \end{align*}
    where the first identity follows from the fact that $(s'_i, s_i)$ and every $\bar{a} \in \bar{A}$ share exactly $1$ unit of flow in $x'$ on a common path (the one on the path with the node sequence $s'$-$s'_1$-$s_1$-$s'_2$-$s_2$-$\dots$-$s'_m$-$s_m$-$\bar{s}$-$t'$ ending with $\bar{a})$.
    Because $\lambda(x', S) \leq kM - 1$ we obtain $\gamma_{ai} \geq 1$ for all $i \in K'$ and $a \in A_i$.

    Note further that, because the arcs $(s_i, s'_{i+1})$ are all saturated by flow on the paths $s'$-$s'_1$-$s_1$-$s'_2$-$s_2$-$\dots$-$s'_m$-$s_m$-$\bar{s}$-$t'$, any flow path of $x'$ that contains both the arcs $(s_i, s'_{i+1})$ and $a \in A_i$ must enter $A$ at $s_i$ and contain an $s_i$-$t_i$-path in $D$ (recall that $\delta_{D}^-(s_j) = \emptyset = \delta_{D}^+(t_j)$ by property~\ref{prop:commodity-source-sink}, and if an $s'$-$t'$-path in $D'$ contains any arc of $A$, it must contain an $s_i$-$t_j$-path for some $i, j \in K$).
    We conclude that $x$ as defined in the statement of the lemma sends $\sum_{a \in A_i} \gamma_{ai} \geq d_i$ units of flow from $s_i$ to $t_i$ for each $i \in K'$. Because $\sum_{a \in \delta_D^+(s_i)} u_a = \sum_{a \in \delta_D^-(t_i)} u_a = d_i$, there must in fact be exactly $d_i$ units of flow from $s_i$ to $t_i$ in $x$.
    
    Further note that any path containing $(s', s_0)$ contains an arc of $A$ and must thus contain an $s_0$-$t_j$-path in $D$. 
    Because the in-capacity $\sum_{a \in \delta_D^-(t_j)} u_a = d_j$ of any sink $t_j$ for $j \neq 0$ is already saturated by the $d_j$ units of flow on $s_j$-$t_j$-paths, all $x'[(s', s_0)] = d_0$ units of flow on $(s', s_0)$ must be send along paths that contain an $s_0$-$t_0$-path in $D$ (together with $(s', s_0)$ and $(t_0, t')$).
    We conclude that $x \in X_K$.
    
    Now consider any $S \in \mathcal{S}_k$ and note that because $A \subseteq A'$ also $S \in \mathcal{S}'_k$.
    This implies $\lambda(x, S) = \lambda(x', S) \leq kM-1$, completing the proof of the lemma. \qed
\end{proof}

\end{document}